\documentclass[letterpaper,11pt,leqno]{article}
\usepackage{paper}
\usepackage{math}

\usepackage{algorithm}
\usepackage{algpseudocode}

\defcitealias{abowd1999high}{AKM}
\defcitealias{bonhomme2019distributional}{BLM}

\hypersetup{pdftitle={TWICE: Tree-based Wage Inference with Clustering and Estimation}}

\newcommand{\bib}{bibliography.bib}

\begin{document}

\title{TWICE: Tree-based Wage Inference with Clustering and Estimation}

\author{Aslan Bakirov, Francesco Del Prato, Paolo Zacchia
%
\thanks{Aslan Bakirov: \emph{CERGE-EI}, \href{mailto:aslan.bakirov@cerge-ei.cz}{\texttt{aslan.bakirov@cerge-ei.cz}}. 
Francesco Del Prato: \emph{Department of Economics and Business Economics, Aarhus University}, \href{mailto:francesco.delprato@econ.au.dk}{\texttt{francesco.delprato@econ.au.dk}}. 
Paolo Zacchia: \emph{Ca' Foscari University of Venice}, \emph{CERGE-EI}, and \emph{CEPR}, \href{mailto:paolo.zacchia@cerge-ei.cz}{\texttt{paolo.zacchia@cerge-ei.cz}}.
This paper results from research funded under the umbrella of the ERC-CZ project no. LL2319.
}}

\date{January 2026}


\begin{titlepage}
\maketitle


How much do worker skills, firm pay policies, and their interaction contribute to wage inequality? 
Standard approaches rely on latent fixed effects identified through worker mobility, but sparse networks inflate variance estimates, additivity assumptions rule out complementarities, and the resulting decompositions lack interpretability. 
We propose TWICE—Tree-based Wage Inference with Clustering and Estimation—a framework that models the conditional wage function directly from observables using gradient-boosted trees, replacing latent effects with interpretable, observable-anchored partitions. 
This trades off the ability to capture idiosyncratic unobservables for robustness to sampling noise and out-of-sample portability. 
Applied to Portuguese administrative data, TWICE outperforms linear benchmarks out of sample and reveals that sorting and non-additive interactions explain substantially more wage dispersion than implied by standard AKM estimates.

\end{titlepage}

\section{Introduction}\label{s:introduction}

A central question in labor economics is to what extent earnings inequality arises from differences among workers, the firms that employ them, and the patterns of sorting between the two. 
The seminal framework of \cite{abowd1999high} (AKM) provided a powerful empirical tool for decomposing wages into worker and firm fixed effects, revealing that firms account for a substantial share of wage variation. 
For over two decades, this two-way fixed effects model has been the workhorse for studying the sources of inequality, the firm-size wage premium, and the nature of labor market sorting.

Recent work, however, has identified core challenges to this paradigm. 
The first concerns “limited mobility bias”: 
because firm effects are solely identified from workers who move between employers, sparse mobility networks can lead to noisy estimates that overstate firm heterogeneity while understating sorting \citep{andrews2008high, kline2020leave, bonhomme2023much}.
The second concerns the model's restrictive additivity assumption. 
\cite{bonhomme2019distributional} (BLM) address both issues by grouping firms into a finite number of latent classes inferred from earnings distributions, allowing for unrestricted worker-firm interactions.
They show that sorting estimates are large and that firm effects are modest. 
Yet because BLM's classes are latent, they can blur the distinction between wage-setting policies and workforce composition, and the results remain inherently in-sample—offering limited guidance on why certain firms pay more or what observable characteristics drive sorting.

\paragraph{This paper}
This paper introduces \emph{TWICE}—Tree-based Wage Inference with Clustering and Estimation—a framework that models the conditional expectation of wages directly from measurable worker and firm characteristics rather than from latent effects identified through mobility. 
Worker observables include education, tenure, age, and occupation; firm observables include size, productivity, financial structure, and sector. 
This shift trades the ability to capture purely idiosyncratic unobservables for three practical advantages: robustness to the sampling noise that plagues AKM estimates in sparse networks; out-of-sample portability to new workers and firms; and direct interpretability of how specific characteristics shape wages. 
Despite relying on flexible machine learning methods, TWICE retains the two-way decomposition structure of the AKM tradition and produces directly comparable outputs: variance shares for worker, firm, sorting, and interaction components.

\paragraph{TWICE}
The TWICE procedure has three steps. 
First, we partition workers and firms into discrete cells based on observable attributes using supervised decision trees—yielding interpretable groups such as “mid-tenure, tertiary-educated technicians” or “large, solvent manufacturers.” 
Second, we estimate the conditional wage function using gradient-boosted trees with cross-fitting that respects the two-sided dependence structure of matched data: workers and firms are split into blocks, and each observation is predicted by a model trained on data excluding that worker and firm \citep{chernozhukov2018double, chiang2022multiway}. 
Third, from this estimated function we extract variance decompositions, sorting patterns, and Partial Dependence and Accumulated Local Effect plots that show how individual characteristics shape predicted wages.

\paragraph{Applications}
Applying TWICE to comprehensive Portuguese administrative data, we document three main findings. 
First, the model achieves substantially better out-of-sample predictive performance than linear benchmarks (test $R^2 \approx 0.50$), indicating that it captures structural features of the wage-setting process rather than noise. 
Second, our variance decomposition—based on observable-anchored cells rather than latent effects—reveals a strong role for sorting on observables (11.6\% of wage variance) and modest but nonzero non-additive interactions (7.3\%), compared to 7.2\% sorting and no interaction term in the standard AKM specification on the same data. 
This pattern is consistent with the bias-corrected results of \citet{bonhomme2019distributional} and the latent-class analysis of BLM: once correlated heterogeneity is modeled flexibly, sorting accounts for more wage dispersion than the additive AKM model suggests, while pure firm effects are smaller. 
Third, we find that our observable-based firm classes explain approximately 25\% of the variation in AKM firm effects, confirming that TWICE captures a meaningful component of firm pay premia while remaining agnostic about idiosyncratic variation—much of it estimation error—that dominates AKM estimates in sparse networks.

\paragraph{Interpretability}
A key contribution of TWICE is interpretability. 
To open the “black box” of machine learning, we apply Partial Dependence Plots \citep{friedman2001greedy, goldstein2015peeking} and Accumulated Local Effects \citep{apley2020visualizing}—diagnostic tools standard in machine learning but, to our knowledge, novel in wage decomposition.
These diagnostics recover economically meaningful patterns: age--wage profiles are concave; returns to tenure vary by worker qualification; wages rise with firm productivity, consistent with rent-sharing. 
Most strikingly, the canonical firm-size wage premium vanishes once worker and firm observables are flexibly controlled, suggesting that the premium reflects sorting—larger firms employ higher-skilled workers and have other high-wage characteristics—rather than a direct effect of size.

\paragraph{Contribution to the literature}
Our work makes three principal contributions to the literature on two-way heterogeneity and wage dispersion.

First, we contribute to the literature on limited mobility bias in AKM estimates.
Standard fixed-effects estimators suffer from an incidental parameter problem that inflates firm-effect variance and attenuates sorting estimates \citep{andrews2008high, kline2020leave, bonhomme2023much, kline2024firm}. 
Our observable-anchored approach sidesteps this mechanism: because we estimate group-level conditional means rather than individual fixed effects, the bias does not apply.
Consistent with this interpretation, our variance decomposition yields a sorting share substantially larger than the uncorrected AKM estimate on the same data.

Second, relative to the latent-class framework of \citet{bonhomme2019distributional}, we replace latent classes with observable-anchored partitions. 
BLM clusters firms using earnings distributions, which can conflate wage-setting policies with workforce composition; our partitions are functions of measured firm and worker attributes, making them directly interpretable. 
Because these partitions are defined by observables rather than by sample identifiers, they extend naturally to new workers and firms—enabling out-of-sample prediction and policy-relevant counterfactuals that latent-class or fixed-effect methods cannot readily support.
Despite this shift, we retain AKM-style outputs—variance shares, sorting matrices—facilitating comparison with the existing literature.

Third, we bring modern machine learning to matched employer--employee data in a transparent way. 
We adapt the cross-fitting principle of \citet{chernozhukov2018double} to the two-way dependence structure of matched data \citep{chiang2022multiway}, and we interpret the resulting estimates using Partial Dependence Plots and Accumulated Local Effects. 
Our findings connect to a long literature on firm contributions to inequality \citep{card2013workplace, barth2016s, song2019firming}: once observables and their interactions are modeled flexibly, sorting and non-additive complementarities account for a larger share of wage dispersion than standard methods suggest.

\paragraph{Outline}
The remainder of the paper is organized as follows.
Section \ref{s:methodology} develops the framework.
Section \ref{s:results} presents empirical results using Portuguese data.
Last, section \ref{s:conclusions} concludes.

\section{A predictive two-way wage structure with observable-anchored partitions}\label{s:methodology}

We develop a framework that models wages flexibly as a function of rich worker and firm observables while summarizing two-sided heterogeneity through feature-anchored partitions (cells) for workers and firms.

\subsection{Objective and approach}

Our goal is to model the conditional expectation of wages as a flexible function of observable worker and firm attributes.
From this estimated function, we extract interpretable summaries—variance decompositions, sorting patterns, and marginal effects—that parallel the outputs of the AKM tradition while being anchored in observables rather than latent effects.

This objective differs from the standard approach in the literature. 
The AKM model expresses wages as the sum of latent worker and firm fixed effects plus time-varying controls. 
Identification relies on exogenous mobility and additivity, and firm effects are identified only within the largest connected set. 
AKM has been central to measuring two-way heterogeneity, but it treats observables as residual controls and does not map them explicitly into wage premia—limiting both interpretability and out-of-sample use.

The BLM model reduces dimensionality by clustering firms into latent classes and estimating a finite mixture of worker types conditional on these classes. 
This relaxes additivity and, by coarsening the firm space, mitigates limited-mobility bias.
However, latent classes remain functions of earnings distributions rather than of observable firm characteristics, again limiting interpretability and portability to new firms.\footnote{%
Appendix \ref{a:AKM_BLM} provides a deeper review.
For a comprehensive treatment, see \citet{kline2024firm}.} 

We take a complementary approach: estimate the conditional expectation function nonparametrically, summarize it through observable-anchored partitions, and retain the two-way decomposition structure while gaining interpretability and out-of-sample portability.

\subsection{Population estimands and decomposition}\label{ss:population_estimands}

Before detailing the estimation procedure, we formally define the population objects of interest. 
Let $Y_{it}$ denote the log wage of worker $i$ at time $t$, employed at firm $j=J(i,t)$. 
Let $Z_{it}$ and $X_j$ denote the vectors of observable characteristics for the worker and the firm, respectively.

We assume the data is generated by a conditional expectation function (CEF) $m_0$, such that:
\begin{equation*}
    Y_{it} = m_0(X_j, Z_{it}) + u_{it}
\end{equation*}
where $\E{u_{it} \mid X_j, Z_{it}} = 0$. 
The function $m_0(\cdot)$ captures the mapping from observables to wages in the population. 
Our goal is to approximate this function and decompose the resulting variation into components attributable to worker types, firm types, and their interactions.

To reduce dimensionality and facilitate economic interpretation, we summarize the continuous function $m_0(\cdot)$ by projecting it onto discrete worker and firm partitions. 
Fix integers $L$ and $K$, and let $g:\mathcal{Z}\to\{1,\dots,L\}$ and $h:\mathcal{X}\to\{1,\dots,K\}$ denote candidate worker and firm partition functions. 
For any pair $(g,h)$, define the corresponding cell means
\[
\mu_{\ell k}(g,h) \equiv \E{Y_{it} \mid g(Z_{it}) = \ell, h(X_j) = k}
\]
We then define the \emph{L$_2$-optimal population partitions} $(g^*,h^*)$ as the solution to
\begin{equation}\label{eq:argmin}
    (g^*, h^*) = \arg\min_{g, h} \E{ \bp{ m_0(X_j, Z_{it}) - \mu_{g(Z_{it}), h(X_j)} (g,h) }^2 }
\end{equation}
That is, $(g^*,h^*)$ yields the best step-function approximation to the true conditional expectation $m_0$ among all partitions with $L$ worker groups and $K$ firm groups.\footnote{%
We do not assume that wages are truly stepwise in $(g^*,h^*)$; instead, these partitions are population summaries of $m_0$, which TWICE approximates using tree-based methods in finite samples.}

For notational convenience, write
\[
\mu_{\ell k} \equiv \mu_{\ell k}(g^*,h^*)
= \E{Y_{it} \mid g^*(Z_{it}) = \ell, h^*(X_j)=k},
\]
and let $\mu \equiv \E{Y_{it}}$ denote the grand mean. 
 We decompose each cell mean into worker, firm, and interaction components as
\[
\mu_{\ell k} = \mu + \alpha_\ell + \psi_k + \kappa_{\ell k}
\]
where the worker and firm premia are defined via an additive projection to ensure an orthogonal decomposition.
Specifically, $(\alpha^*, \psi^*)$ solve the weighted least-squares problem
\[
    (\alpha^*, \psi^*) = \arg \min_{\alpha,\psi} \sum_{\ell, k} \pi_{\ell k} \bp{\mu_{\ell k} - \mu - \alpha_\ell - \psi_k}^2
\]
where $\pi_{\ell k} = \Pr{g^*(Z_{it}) = \ell, h^*(X_{j}) = k}$.
This is equivalent to regressing cell means on worker-group and firm-group indicators, weighted by cell sizes.
By construction, $\sum_\ell \pi_\ell \alpha_\ell^* = 0$ and $\sum_k \pi_k \psi_k^* = 0$.

The key property of this projection is orthogonality: the interaction term 
\[
\kappa_{\ell k} = \mu_{\ell k} - \mu - \alpha^*_\ell - \psi^*_k
\]
which captures deviations from additivity at the cell level, is uncorrelated with both $\alpha^*$ and $\psi^*$ by construction, ensuring that the variance components sum exactly to the total without cross-terms.\footnote{%
When the distribution of workers across firms is balanced ($\pi_{\ell k} = \pi_\ell \pi_k$, the projected effects coincide with the conditional means $\E{\mu_{\ell k} \mid g = \ell}$. 
Under sorting, however, these differ: the projection isolates the “pure” worker and firm  contributions from their correlation induced by assortative matching.
}
When $\kappa_{\ell k}=0$ for all $(\ell,k)$, the systematic component of wages is additive in worker and firm types; nonzero values of $\kappa_{\ell k}$ reflect complementarities or mismatch that are specific to particular worker–firm type pairs.

This yields the following variance decomposition of log wages:
\begin{equation*}
    \begin{split}
    \var{Y_{it}} 
    = \underbrace{\var{\alpha_{g^*(Z_{it})}}}_{\text{worker}} 
    + &\underbrace{\var{\psi_{h^*(X_j)}}}_{\text{firm}} 
    + \underbrace{2 \cov{\alpha_{g^*(Z_{it})}, \psi_{h^*(X_j)}}}_{\text{sorting}}  \\
    + &\underbrace{\var{\kappa_{g^*(Z_{it}), h^*(X_j)}}}_{\text{interaction}} 
    + \underbrace{\var{\xi_{it}}}_{\text{residual}}
    \end{split}
\end{equation*}
The residual term is defined as
\[
\xi_{it} = Y_{it} - \mu_{g^*(Z_{it}), h^*(X_j)}
\]
so that $\E{\xi_{it} \mid g^*(Z_{it}), h^*(X_j)} = 0$. 
It captures three distinct sources of variation: \emph{i)} the pure stochastic error $u_{it}$,
\emph{ii)} idiosyncratic match effects not captured by the group-level means,
and \emph{iii)} the approximation error arising from discretizing $m_0(\cdot)$ into finite cells (i.e., within-cell heterogeneity).
By construction of the additive projection, $\cov{\alpha^*, \kappa^*} = \cov{\psi^*, \kappa^*} = 0$, so no additional cross-terms appear. 
This contrasts with a naive decomposition using marginal means, where cross-covariances between marginal effects and interactions are generally nonzero under sorting and must be accounted for separately.

The TWICE procedure, described below, estimates these population objects by replacing the theoretical minimization in eq. \eqref{eq:argmin} with gradient-boosted decision trees, and estimating the cell means $\mu_{\ell k}$ via cross-fitted regression.
Appendix \ref{a:theory} provides the formal derivation of these components from cell means and the consistency conditions for the estimator.

\subsection{Estimation}\label{ss:estimation}

We estimate the conditional wage function $m_0$ nonparametrically using gradient-boosted trees. 
This choice is driven by two features of wage determination and one of the data structure.

\begin{itemize}
    \item \emph{Nonlinearity and interactions}.
    Returns to worker characteristics (e.g., education, experience, occupation) often depend on firm attributes (e.g., size, productivity, solvency). 
    Trees capture such interactions without requiring pre-specified functional forms or cross-terms.
    
    \item \emph{Regularization and generalization}.
    Boosted trees incorporate shrinkage, depth constraints, and early stopping, which control model complexity and deliver stable out-of-sample predictions—essential for interpreting the fitted structure through worker and firm partitions.

    \item \emph{Dependence in matched employer–employee data}.
    Observations sharing a worker or firm identifier are strongly dependent:
    standard cross-validation would be invalid because the same identifier could appear in both training and test sets.
    We address this using two-way identifier-blocked cross-fitting: workers and firms are each partitioned into disjoint blocks, and each observation is predicted by a model trained on data that excludes the block containing that worker and the block containing that firm \citep{chernozhukov2018double, chiang2022multiway}.
    This yields unbiased out-of-sample fit measures despite the multiway dependence in matched data.\footnote{%
    Dependence may be arbitrarily strong within worker and firm IDs but is assumed weak across distinct IDs; see Appendix \ref{a:theory}.}
\end{itemize}
\medskip

Because the fitted model conditions directly on observables, its predictions can be decomposed into contributions associated with specific measurable traits. 
We examine these relationships in Sections \ref{ss:pdp} and \ref{ss:ale} using Partial Dependence and Accumulated Local Effects plots, which characterize the fitted associations without imposing parametric restrictions.

\subsection{TWICE}\label{ss:twice}
We propose \emph{Tree-based Wage Inference with Clustering and Estimation} (TWICE), a procedure that estimates the conditional wage function using observables and summarizes two-sided heterogeneity through \emph{observable-anchored partitions} (cells) for workers and firms. 
Let
\[
Y_{it} = m_0(X_j,Z_{it}) + u_{it}
\]
denote the data-generating process in Section~\ref{ss:population_estimands}, and let $f$ denote the cross-fitted gradient-boosted tree estimator of $m_0$ described in Section~\ref{ss:estimation}.
TWICE consists of three steps.

\paragraph{I. Firm- and worker-cells from observables} 
We first build discrete, interpretable partitions (cells) from observables.

On the firm side, we estimate a supervised tree regression of firm-level mean log wages $Y_j$ on firm covariates $X_j$ and use the terminal leaves of a selected tree as \emph{firm-cells}. 
These cells group firms with similar wage-setting behavior conditional on observables (e.g., “large solvent manufacturers”).

On the worker side, we use the full panel and estimate a tree regression of individual log wages $Y_{it}$ on time-varying worker covariates $Z_{it}$ (age, tenure, education, occupation, etc.). 
The terminal leaves of this tree define \emph{worker-cells} at the observation level: each $(i,t)$ is assigned to a cell $g(Z_{it})\in\{1,\dots,L\}$ according to its observable characteristics. In the variance decomposition, worker components are evaluated at these observation-specific cells $g(Z_{it})$.

For both firm and worker trees, the number of leaves (depth) is chosen by out-of-sample performance. 
Because leaves are functions of observables, cells are directly interpretable (e.g., “mid-tenure tertiary-educated technicians”).

\paragraph{II. Wage predictor with two-way cross-fitting}
We then fit a LightGBM predictor for $m_0(\cdot)$ using the full observable set and the learned cell indicators. 
To avoid overfitting and obtain unbiased out-of-sample predictions, we implement two-way ID-blocked cross-fitting, following \cite{chernozhukov2018double} and \citet{chiang2022multiway}.
We partition worker IDs into $B$ blocks and firm IDs into $B$ blocks, form $B^2$ validation cells $S_{ab}$. 
For each $(a,b)$, we train on the complement of $S_{ab}$ and predict only for $S_{ab}$, so that no worker or firm appears in both training and validation for its own prediction. 
We select hyperparameters, including the number of worker and firm cells (granularity), by minimizing this blocked out-of-sample risk (see Appendix \ref{a:twice} for details).
All folds and holdouts are constructed within the largest connected set of the mobility graph to maintain comparability with AKM.
Further details on cross-fitting and implementation appear in Subsection~\ref{ss:implementation}.

\paragraph{III. Descriptive outputs: sorting and variance decomposition} 
Using the worker- and firm-cells, we compute the sorting matrix (shares of worker-cells across firm-cells) and decompose log-wage variance into worker-cell, firm-cell, sorting, interaction (cell non-additivity), and within-cell components, as in Section~\ref{ss:population_estimands}. 
For interpretability, we probe $f(\cdot)$ with Partial Dependence and Accumulated Local Effects to summarize how predicted wages vary with key firm features (size, revenue per worker, solvency, capital intensity) and worker features (age, tenure, education/qualification).

\paragraph{Advantages}
Our approach assumes that the salient heterogeneity in wages can be well approximated using observable worker and firm attributes. 
It is conceptually related to BLM’s idea of coarsening the firm space, but differs in three important ways.

First, the partitions are \emph{observable-anchored and transparent}. 
As discussed in Section \ref{ss:estimation}, cells correspond to explicit combinations of measurable traits, enabling direct economic interpretation.

Second, the clustering is \emph{conditional and supervised}. 
BLM classifies firms by (residualized) earnings distributions; class formation does not directly incorporate rich firm covariates.
We instead let observables drive the partitions: cells are learned from $X_j$ (for firms) and $Z_{i\cdot}$ (for workers) in a supervised way. 
This reduces the risk of conflating workforce composition with wage-setting and yields partitions that can be applied to new firms and workers.

Third, the framework is designed for \emph{out-of-sample validity}. 
By combining boosted trees with ID-blocked cross-fitting, TWICE can handle high-dimensional covariates while controlling overfitting in matched employer–employee data.

\subsection{Interpretation and identification}\label{ss:identification}
While TWICE shares the two-way variance decomposition structure of AKM, the identification of the components relies on different sources of variation. 
The standard AKM model identifies firm effects via worker mobility, relying on the strict exogeneity of moves conditional on person and firm fixed effects. 
In contrast, TWICE identifies firm and worker contributions based on the mapping between realized wages and observable characteristics ($X_j, Z_{it}$).

To formalize the contrast, consider the latent AKM firm effect $\psi_j$. 
This parameter captures the time-invariant firm-specific wage premia common to all workers at the firm $j$, regardless of its source. 
We can decompose it as
\[
\psi_j = \phi(X_j) + \nu_j
\]
where $\phi(X_j)$ represents the component that is systematically associated with observable firm characteristics, and $\nu_j$ is a residual term.
This residual bundles together: 
(i) genuine unobserved firm heterogeneity, and 
(ii) sampling noise inherent in estimating firm effects from sparse mobility data.

This decomposition highlights the identification trade-off between AKM and TWICE.
AKM target the total firm premium $\psi_j$ and, in principle, recovers both $\phi(X_j)$ and $\nu_j$.
However, when mobility is limited, the estimated dispersion of firm effects and the estimated worker–firm covariance are affected by sampling variability: the variance of $\widehat\psi_j$ is inflated and the covariance between worker and firm effects is attenuated \citep{andrews2008high, bonhomme2019distributional}.
    
TWICE, by contrast, does not use the mobility graph to identify firm premia.
It estimates the conditional expectation function $m_0(X_j,Z_{it})$ and then summarizes it through observable-based partitions $(g^*,h^*)$.
The resulting firm component in the variance decomposition captures the part of wage premia that is predictable from observables $X_j$.
Conceptually, this corresponds to $\phi(X_j)$ rather than $\psi_j$:
idiosyncratic fluctuations and firm-specific factors that are orthogonal to $X_j$ are absorbed into the residual and interaction terms, rather than into the firm component.
Consequently, the firm component in TWICE should be interpreted as the \emph{observable firm wage premium}. 

This distinction extends directly to the interpretation of sorting.
In the AKM framework, sorting is the covariance between total latent worker and firm effects, $2\,\cov{\alpha_i, \psi_{J(i,t)}}$, so it reflects assortative matching on latent heterogeneity.
In TWICE, sorting represents \emph{assortative matching on observables}, $2\,\cov{\alpha_{g^*(Z)}, \psi_{h^*(X)}}$, i.e., the extent to which workers with high-wage observable attributes match with firms possessing high-wage observable attributes.

The cost of this shift is an omitted-variable concern on the firm side.
If high-wage firms pay premia for reasons strictly orthogonal to our observable $X_j$, and these reasons are not indirectly reflected in observable sorting patterns, TWICE will not attribute that variation to the firm component.
Instead, it will appear in the residual or interaction terms.
In our application, however, the administrative data contain rich financial and workforce information at the firm level and detailed qualifications on the worker side.
Existing two-way methods primarily exploit mobility to uncover latent heterogeneity and typically use such observables only as controls.
TWICE is designed to complement that literature by asking how much of wage dispersion and sorting can already be accounted for by this rich observable structure, and by providing a decomposition that is explicitly framed in terms of observable firm wage premia and assortative matching on observables.

\section{An application to Portuguese data}\label{s:results}

In this section, we showcase an application of the TWICE framework to Portuguese data. 
We analyze sorting patterns—the extent to which high-wage workers match with high-wage firms—and compare our findings to BLM. 
Additionally, we show how PDPs and ALEs can display nonparametric relationships between observables in wage determination.

\subsection{Data}

\paragraph{Sources}
We combine two Portuguese administrative datasets from \emph{Statistics Portugal}.
The matched employer-employee database \emph{Quadros de Pessoal} (QP) covers the universe of private-sector firms, providing worker-level information on age, gender, education, occupation, tenure, and earnings, as well as firm-level information on location, industry, and employment.
We supplement QP with firm financial data from \emph{Sistema de Contas Integradas das Empresas} (SCIE), which provides balance-sheet items including revenues, assets, liabilities, and investment.

\paragraph{Sample}
Our analysis spans 2012--2019, providing a recent pre-pandemic window.
We restrict attention to private-sector firms with at least five employees and full-time workers aged 20--65 with at least two months of tenure.
The outcome variable is log hourly wages, constructed from annual earnings and hours worked.
After applying these filters and retaining only observations in the largest connected set (required for AKM comparability), the estimation sample contains approximately 3.4 million worker-year observations, covering 750,000 unique workers across 96,000 firms.

Table~\ref{tab:stat_panel} summarizes the panel.
Employment grows steadily over the period (firms: 43k to 55k; workers: 350k to 507k), while average log wages rise from 1.97 to 2.10.
The workforce gradually upskills: the share of tertiary graduates increases from 17\% to 22\%.
Table~\ref{tab:stat_industry_19} reveals substantial cross-industry heterogeneity: finance pays the highest wages (mean log wage 2.51) and employs the most graduates (72\%), while manufacturing and agriculture sit below the aggregate average.

\paragraph{Covariates}
TWICE uses two sets of observable characteristics.
\emph{Worker covariates} $Z_{it}$ include age, tenure, education, occupation (qualification), and job seniority—time-varying attributes that capture human capital accumulation and job-specific experience.
Notably, we exclude calendar year from the worker clustering stage, so that workers are grouped by where they are in their career rather than when they are observed.
\emph{Firm covariates} $X_j$ include size (log employment), productivity (log revenue per worker), financial structure (solvency ratio, asset composition), investment (R\&D, patents, capital intensity), and market position (concentration index), along with sector, region, legal form, and calendar year.
Appendix~\ref{app:data} provides full variable definitions and descriptive statistics.

\paragraph{Mobility and connectivity}
AKM identification requires a connected mobility network.
In our sample, workers are employed at 2.2 firms on average, and 21\% work at three or more firms during the panel.
The largest connected set retains 99.6\% of workers, indicating a well-connected labor market.
To validate the exogenous mobility assumption underlying the AKM benchmark, Appendix~\ref{app:event_study} replicates the event-study design of \cite{card2013workplace}: wage profiles are flat before job transitions, and gains from moving up the firm distribution are symmetric to losses from moving down, supporting additivity.

\subsection{Implementation}\label{ss:implementation}

This subsection describes how the TWICE framework is implemented in practice, consistent with Section~\ref{ss:twice}.
The procedure consists of three components: (i) constructing partitions of firms and workers based on observables, (ii) estimating the wage function with two-way cross-fitting on the connected set, and (iii) selecting the optimal granularity of the partitions by out-of-sample risk minimization.

\paragraph{Firm classification}
We first construct firm partitions based on observed firm characteristics.  
Let $Y_j$ denote the mean log-wage among employees of firm $j$, and $X_j$ the corresponding firm covariates.  
For a candidate number of groups $K$, we estimate a supervised tree-based model of the form
\[
\widehat f_K = \arg\min_{f\in\Omega_K}\sum_j \bs{Y_j - f(X_j)}^2
\]
where $\Omega_K$ is the class of regression trees with at most $K$ terminal nodes.  
The resulting leaves define firm classes $\kappa_K(j)\in\{1,\dots,K\}$, which summarize heterogeneity in pay-setting conditional on observables such as size, productivity, and sector.  
We interpret each class as a group of firms with similar wage-setting behavior.

\paragraph{Worker classification}
Analogously, we partition workers into $L$ classes based on their observable characteristics and log wages.
For computational efficiency and to reduce noise, we first aggregate to the worker level: let $\bar{Y}_i$ and $\bar{Z}_i$ denote, respectively, the mean log wage and the vector of averaged time-varying covariates for worker $i$, combined with their time-invariant attributes.
We then estimate a supervised tree
\[
\widehat g_L = \arg\min_{g\in\Omega_L}\sum_{i} \bs{\bar{Y}_{i}-g(\bar{Z}_{i})}^2,
\]
where $\Omega_L$ denotes the class of regression trees with at most $L$ leaves.
The resulting splitting rules define an estimated partition function $\widehat g_L:\mathcal Z\to\{1,\dots,L\}$, which we then apply at the observation level: each $(i,t)$ is assigned to group
\[
\ell_{it} = \widehat g_L(Z_{it}) \in \{1,\dots,L\}
\]
based on its period-specific covariates.
Because some components of $Z_{it}$ are time-varying (e.g., age, tenure), the assigned worker class $\ell_{it}$ may change over time for the same individual.

An important design choice concerns the treatment of calendar time in the clustering stage. 
Worker clusters are formed using career-stage characteristics—age, tenure, education, and occupation—but not calendar year. 
This means workers are assigned to clusters based on where they are in their working life rather than when they are observed: a 35-year-old with 10 years of tenure will be in the same worker cluster whether observed in 2012 or 2019. 
Firm clusters, by contrast, incorporate calendar year alongside balance-sheet and structural characteristics.
This allows the same firm to move between clusters as macroeconomic conditions change. 
The asymmetric treatment of time reflects the economic structure of wage determination: individual worker productivity depends primarily on accumulated human capital (career stage), while firm wage policies respond to business cycle conditions and evolving labor market tightness (calendar time).

\paragraph{Wage prediction and model selection}
Given the firm and worker partitions, we estimate the wage function on the largest connected set of the mobility graph using two-way cross-fitting as described in Section~\ref{s:methodology}.  
For each candidate pair $(K,L)$, we compute the blocked out-of-sample loss $\mathcal{L}(K,L)$ defined in Appendix \ref{a:twice} and select
\[
(K^*, L^*) = \arg\min_{K,L} \mathcal{L}(K,L)
\]

In a final step, we re-estimate the model on the training portion of the connected set using $(K^\ast,L^\ast)$ and evaluate its predictive performance on an external firm-level holdout: a random subset of firms, and all their associated worker--year observations, that are excluded from both training and tuning.
This holdout provides an out-of-sample benchmark distinct from the cross-fitting blocks used for model selection.

\subsection{Out-of-sample fit}\label{ss:oos}

We evaluate predictive performance using an external firm-level holdout. 
After restricting to the connected set, we draw a random subset of firms, sample workers within those firms, and exclude all corresponding observations from any model fitting, including cross-fitting and hyperparameter tuning.
All tuning and cross-fitting (Section~\ref{s:methodology}) are conducted on the remaining firms.
Once the tuning parameters are selected, we re-estimate each model on the training firms and compute performance on the held-out firms.
In machine-learning parlance, the cross-fitting folds play the role of a validation sample, while the firm holdout serves as a final test sample.

As baselines, we estimate a sequence of OLS models on the same sample and covariates used by TWICE. 
We report both training and test metrics, with emphasis on the test results. 
The baselines are: 
(i) a “simple” OLS with an age polynomial (normalized to be flat at 40) and year fixed effects; 
(ii) an OLS with the full set of continuous covariates and categorical indicators; and 
(iii) OLS variants that add polynomial expansions of continuous covariates up to degree three.\footnote{%
To avoid numerical instabilities and spurious collinearity, higher-order specifications only apply polynomial expansions to a restricted subset of continuous variables (age, tenure, log-size, log-revenue) while leaving other predictors linear.}
The TWICE model uses the same covariates plus the worker- and firm-cell indicators and is estimated with two-way cross-fitting.

\begin{table}[t]
\caption{Test sample performance of the models}
\label{tab:oos}
\centering
{
\begin{tabular*}{.9\textwidth}[]{p{3cm}@{\extracolsep\fill}cccc}
\toprule
& \multicolumn{2}{c}{Train} & \multicolumn{2}{c}{Test} \\
\cmidrule(lr){2-3} \cmidrule(lr){4-5}
Model & MSE & $R^2$ & MSE & $R^2$ \\
\midrule
OLS & 0.166 & 0.039 & 0.174 & 0.043 \\
OLS, degree-1 poly & 0.098 & 0.431 & 0.106 & 0.415 \\
OLS, degree-2 poly & 0.096 & 0.446 & 0.109 & 0.405 \\
OLS, degree-3 poly & 0.095 & 0.447 & 0.108 & 0.408 \\
TWICE & 0.076 & 0.566 & 0.092 & 0.493 \\
\bottomrule
\end{tabular*}
}

\note[Note]{
This table reports out-of-sample performance metrics for the prediction of log hourly wages on the Portuguese linked employer-employee data (2012–2019). 
The test set is constructed by holding out a random sample of firms (and all their workers) that were not used in model training or hyperparameter tuning. 
MSE is the Mean Squared Error.
$R^2$ is the squared correlation between observed and predicted wages in the test set. 
“OLS” refers to a linear regression with age polynomials and year fixed effects. 
Higher-degree OLS models include polynomial expansions of continuous covariates up to the specified degree. 
“TWICE” refers to the gradient-boosted tree ensemble estimated via two-way cross-fitting.
}
\end{table}

\paragraph{Results}
Table~\ref{tab:oos} compares the out-of-sample predictive performance of the OLS baselines and TWICE.
Introducing polynomial flexibility substantially improves in-sample fit for OLS, but only modestly improves out-of-sample accuracy: test $R^2$ rises from 0.04 to about 0.40 at degree 3.
Even with selective polynomialization and orthogonal bases, linear models reach a plateau in explanatory power, suggesting that key nonlinearities and interactions cannot be captured by additive polynomial terms alone.

In contrast, the TWICE model achieves a markedly better out-of-sample performance, with a test $R^2 \approx 0.50$ and a test MSE roughly 15\% lower than the best OLS variant.
This gap persists despite TWICE’s greater flexibility, underscoring the effectiveness of its tree-based regularization and cross-fitting scheme in balancing flexibility and parsimony.
Unlike the higher-degree OLS models, which improve fit primarily in-sample, TWICE generalizes smoothly to new firms and workers, indicating that it captures structural features of the wage-setting process rather than sampling noise.

The robust generalization of TWICE also enhances the interpretability of the partial dependence plots (PDPs) presented in Subsection~\ref{ss:pdp}.
Because these plots rely on a model’s ability to approximate true conditional expectations, their credibility depends on out-of-sample performance.
The consistent predictive accuracy of TWICE provides confidence that the PDPs reflect genuine wage–covariate relationships rather than overfitted artifacts.

\subsection{Sorting}

We now examine how the model’s estimated worker and firm heterogeneity translate into sorting patterns.
Following the logic of \cite{bonhomme2019distributional}, we analyze how the distribution of worker types varies across firm types, using the cells recovered by TWICE.
Unlike the latent worker and firm types in the BLM framework, these classes are data-driven partitions formed by the TWICE trees based on covariates.
A firm class should be read as a group of firms with similar observable profiles and, therefore, similar predicted wage premia for a generic worker; a worker class collects individuals with similar observable profiles and, hence, similar predicted wage schedules across firms.

\begin{figure}[t]
\caption{Worker–firm sorting patterns}
\label{fig:sorting_plot}
\centering
\includegraphics[width=0.85\textwidth]{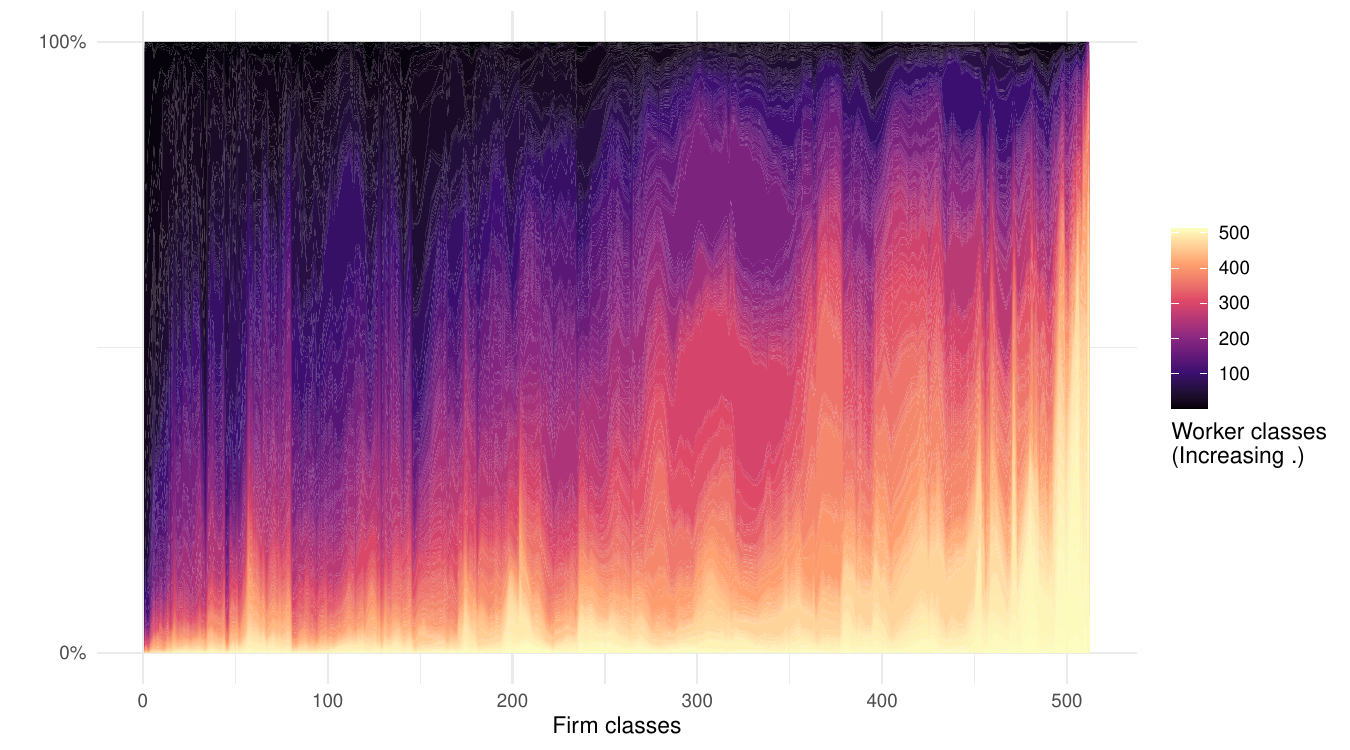}
\note[Note]{
Heatmap showing the joint distribution of worker and firm classes recovered by TWICE. 
The x-axis orders firm classes from lowest to highest average firm premium ($\psi_h$). 
The y-axis represents the composition of the workforce within that firm class. 
Colors represent worker types ordered by average worker effect ($\alpha_g$), from dark (low wage) to light (high wage). 
A diagonal pattern indicates positive assortative matching on observables: high-wage firm types disproportionately employ high-wage worker types.
}
\end{figure}

Figure~\ref{fig:sorting_plot} displays the composition of worker types across the ordered firm spectrum.
The pattern reveals clear evidence of positive assortative matching: low-wage firms (on the left) predominantly employ lower-type workers, while the share of higher-type workers increases smoothly toward the right end of the firm distribution.
This monotonic gradient implies that firm heterogeneity and worker heterogeneity are mutually reinforcing in wage determination—high-paying firms systematically attract and retain more productive workers.

Compared with the discrete-type models of \cite{bonhomme2019distributional}, our implementation uses a relatively fine grid of worker and firm classes.
Tree-based partitioning and LightGBM smoothing make such fine partitions computationally tractable, while the number of leaves is disciplined by out-of-sample risk minimization.
The resulting step-function approximation yields a finely graded, data-driven representation of the joint distribution of worker and firm heterogeneity: rather than a handful of latent classes, we obtain a dense ranking of observable types, where the composition of worker classes evolves gradually along the firm spectrum.
We interpret this as a quasi-continuous hierarchy of matches—an approximation to an underlying smooth distribution of wage premia—consistent with a positive but imperfect correlation between worker and firm effects.
 
\subsection{Partial dependence plots}\label{ss:pdp}

A central advantage of the TWICE framework is its ability to combine strong predictive performance with economically interpretable relationships.
Section~\ref{ss:oos} shows that, on a firm-level holdout sample, TWICE attains substantially higher test $R^2$ and lower test MSE than flexible OLS benchmarks.
This does not turn TWICE into a structural model, but it does suggest that the fitted wage surface captures systematic regularities in the data rather than sample-specific noise.
We use this out-of-sample performance as a credibility check when interpreting model-based summaries such as Partial Dependence Plots (PDPs).

Let $V_{it} = (X_{J(i,t)}, Z_{it})$ denote the full vector of firm and worker covariates for observation $(i,t)$, and let $P_V$ be the joint distribution of $V_{it}$.
For a subset of coordinates $S$ (for example, age or tenure), write $V_{S}$ for the corresponding subvector and $V_{-S}$ for the remaining coordinates.
Given the cross-fitted predictor $\widehat f$ from Section~\ref{ss:estimation}, the partial dependence function for $V_S$ is defined as
\[
    f_S(v_S)
    = \E[V_{-S}]{\widehat f(v_S, V_{-S})}
    = \int \widehat f(v_S, v_{-S}) \, dP_{V_{-S}}(v_{-S})
\]
where $P_{V_{-S}}$ is the marginal distribution of $V_{-S}$ under $P_V$.

The function $f_S(v_S)$ can be read as the average predicted log wage at feature value $v_S$, integrating over the empirical distribution of all other covariates.
It therefore provides a nonparametric, model-based \emph{ceteris paribus} profile: it traces how predicted wages vary with a given attribute while holding the rest of the feature vector at its observed distribution, without imposing a pre-specified functional form.
In what follows, we interpret these PDPs as descriptive summaries of the estimated conditional wage function, not as causal responses to changes in the underlying covariates.

In our empirical implementation, we report a \emph{reference-point} version that fixes non-focal covariates at representative values.
Let $\tilde v_{-S}$ collect the sample medians of non-focal continuous covariates and the sample modes of non-focal categorical covariates (computed on the estimation sample).
For any displayed subgroup $d$ (e.g. qualification $\times$ gender $\times$ education), let $D$ denote the set of subgroup indicators and define
\[
\tilde f^{S}(v_S;d) = \widehat f\of{v_S,d,\tilde v_{-(S\cup D)}}.
\]
This yields a ceteris paribus profile evaluated at a common baseline. 
Appendix \ref{app:pdps} provides full implementation details (grid construction, trimming, and cross-fitted averaging).
As throughout, these plots are descriptive summaries of the fitted conditional wage surface, not causal effects.

In the next subsections, we present PDPs for key determinants such as worker age, tenure, and firm size to highlight the nonlinear and interactive structure of wage formation uncovered by TWICE.

\paragraph{Age, gender, and education, by qualification}
Figure~\ref{fig:pdp_age_sex} presents PDPs of predicted log wages by age, conditional on the worker qualification classes identified in the TWICE clustering stage. 
Each panel corresponds to a qualification type—generic worker, specialized worker, and manager—and plots age–wage profiles separately by gender and across three education levels: at most primary, secondary, and at least a bachelor’s degree.

\begin{figure}[p]
\centering
\caption{PDP for age, gender, and education – by qualification}
\label{fig:pdp_age_sex}
 \subcaptionbox{Generic worker\label{fig:pdp_generic}}{\includegraphics[scale=0.59]{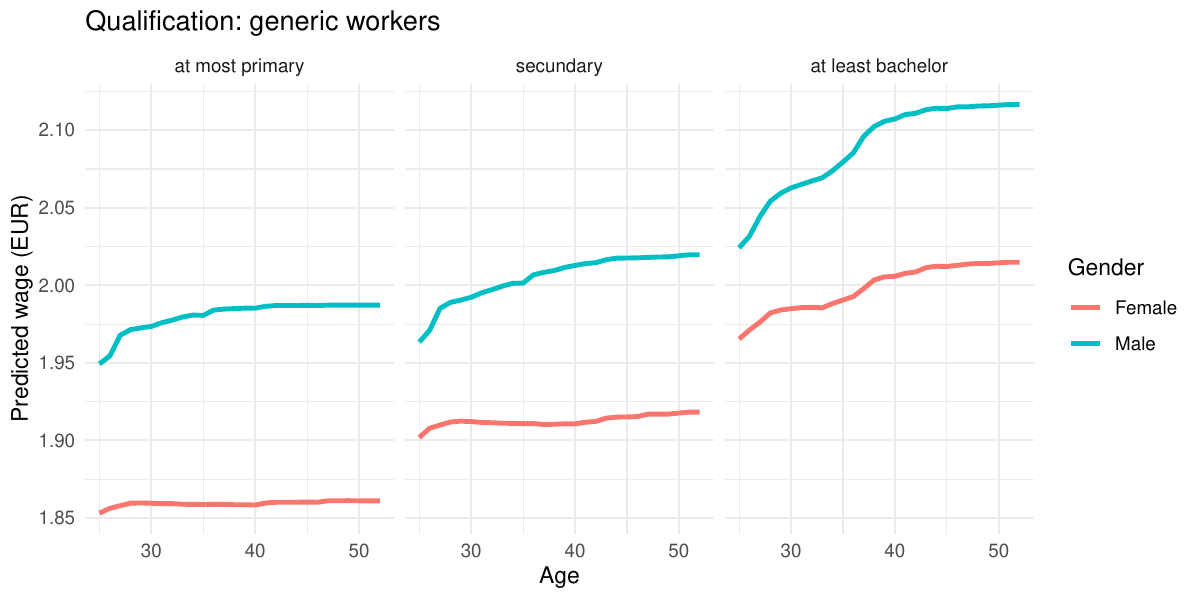}}\\
\subcaptionbox{Specialized worker\label{fig:pdp_specialized}}{\includegraphics[scale=0.59]{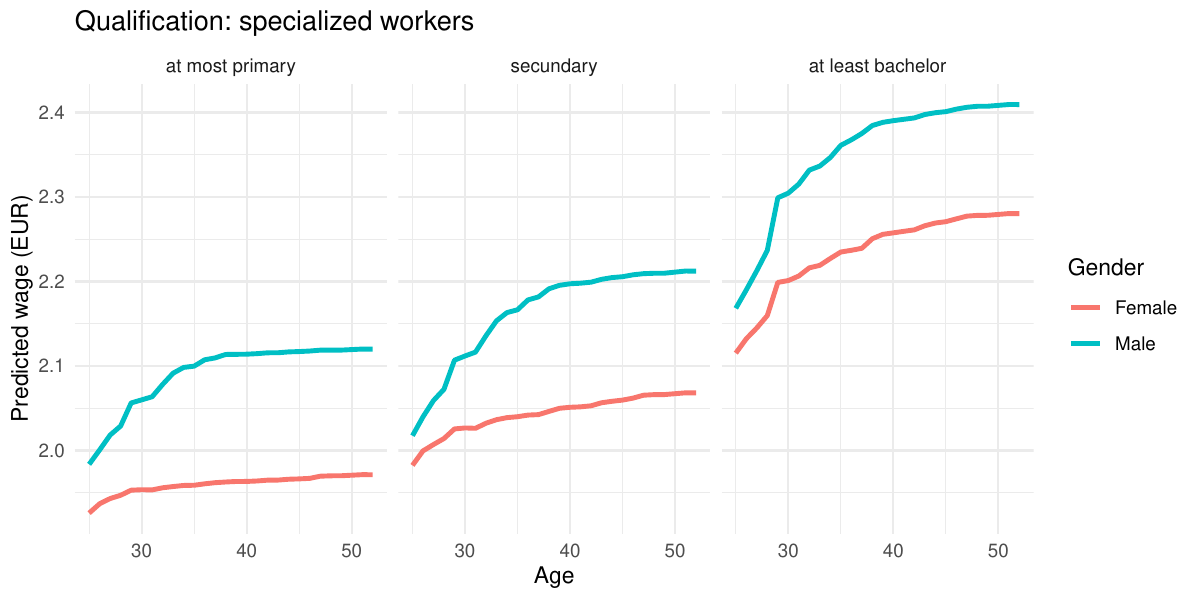}}\\
\subcaptionbox{Manager\label{fig:pdp_manager}}{\includegraphics[scale=0.59]{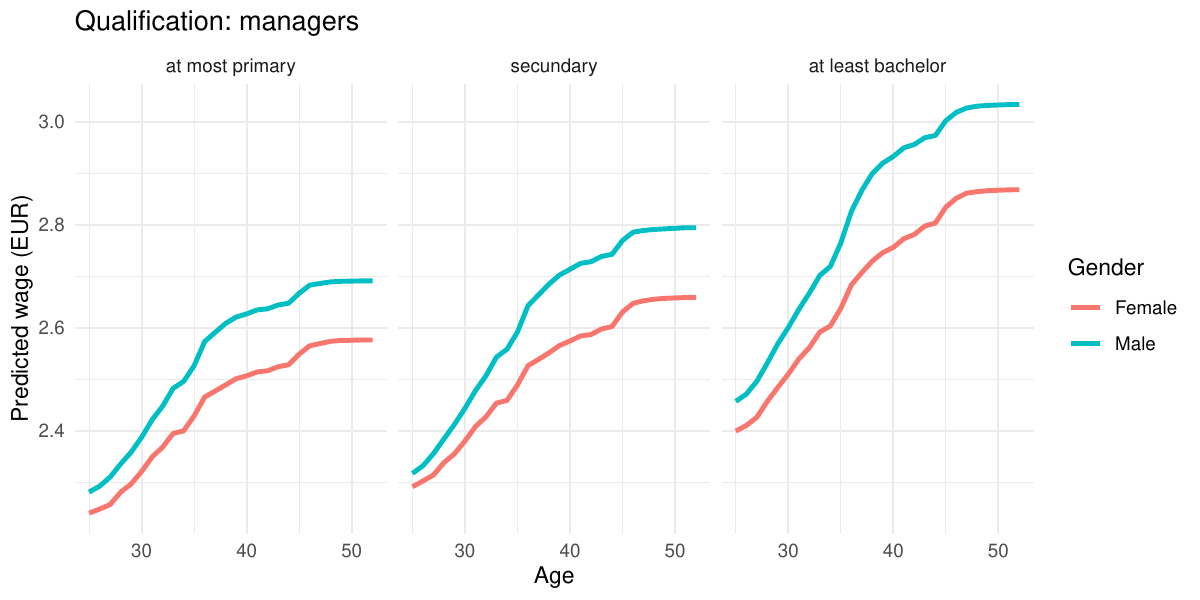}}

\note[Note]{
Partial Dependence Plots showing the predicted log hourly wage (y-axis) as a function of age (x-axis). 
Panels are stratified by the worker qualification level determined by the model (Generic, Specialized, Manager).
Lines are separated by gender (color) and education level (facets). 
Predictions represent the marginal effect of age, averaging over the empirical distribution of all other worker and firm characteristics in the model.
}
\end{figure}

Conditioning on qualification reveals substantial heterogeneity in wage dynamics that would be obscured in a pooled analysis. 

For \emph{generic workers} (Figure~\ref{fig:pdp_generic}), wage profiles are moderately concave, with returns to experience flattening after the mid-forties, consistent with standard human-capital accumulation \citep{mincer1974schooling, card2016bargaining}. 
Education premiums are visible and stable across the life cycle, particularly for tertiary education. 
The slope of the age–wage profile is similar for men and women, though a persistent level gap remains throughout the working life.

Among \emph{specialized workers} (Figure~\ref{fig:pdp_specialized}), the model captures steeper early-career wage growth, especially for degree holders. 
The gender wage gap widens noticeably after the late thirties, coinciding with the prime career progression phase and consistent with the literature on gender differences in promotion and job mobility \citep[e.g.][]{bertrand2013career}. 
For less-educated workers, the profiles are flatter and converge at lower wage levels, indicating more limited returns to experience.

For \emph{managers} (Figure~\ref{fig:pdp_manager}), age–wage profiles are markedly steeper and continue rising into the fifties, suggesting sustained returns to tenure and firm-specific human capital \citep{baker1994internal,lazear2018compensation}. 
The observed stepwise structure likely reflects the tree-based partitioning of TWICE, capturing discrete wage thresholds associated with promotions or responsibility levels—features that smooth parametric specifications would obscure. 
The divergence between men and women is largest in this group: while early-career trajectories overlap, male managers’ wages rise faster and for longer, producing the widest absolute gender gap at older ages and higher education levels.

\paragraph{Tenure, education, and qualification}
We next investigate the returns to firm-specific human capital by examining the relationship between tenure and predicted wages, as shown in Figure~\ref{fig:pdp_tenure_edu}.
The figure plots predicted log wages against months of tenure (on a logarithmic scale), conditioning on worker qualification and education level.
This allows us to assess how the value of firm experience differs across worker types—a central question in the literature on human capital and internal labor markets \citep{topel1991specific}.

\begin{figure}[t]
\centering
\includegraphics[width=0.9\textwidth]{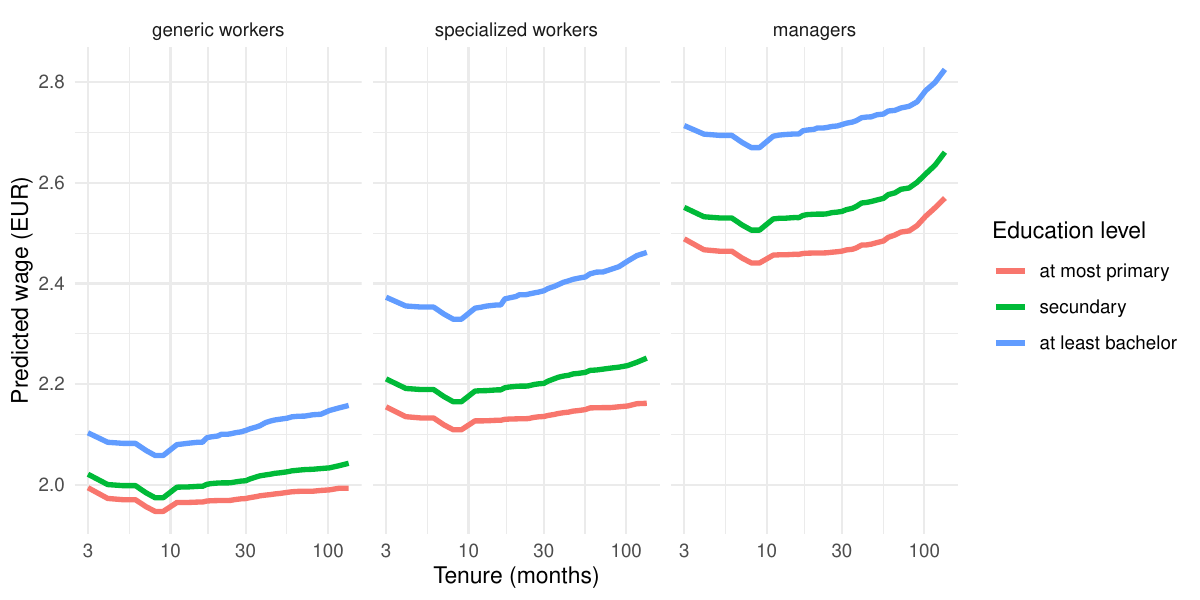}
\caption{PDP for tenure, education, and qualification}
\note[Note]{Partial Dependence Plots derived from the TWICE model. 
Each panel illustrates the marginal relationship between predicted log-wage (y-axis) and tenure in months (x-axis, log scale), conditional on education level (color) and worker qualification (facet). 
Predictions are averaged over the empirical distribution of all other worker and firm characteristics.}
\label{fig:pdp_tenure_edu}
\end{figure}

The results reveal substantial heterogeneity in the returns to tenure.
For \emph{generic workers}, wage profiles are nearly flat, suggesting limited firm-specific learning and a predominance of general, transferable skills.
For \emph{specialized workers}, the profiles become steeper, indicating that firm-specific skills and experience play a larger role in pay growth.
Among \emph{managers}, the relationship between tenure and predicted wages is clearly positive and concave, with steep early gains that taper gradually—consistent with models of internal promotion and long-term incentive contracts \citep{baker1994internal,lazear2018compensation}.

A notable feature across all qualification groups is a short initial dip in predicted wages for workers with less than one year of tenure.
A simple probationary-wage story could account for lower initial pay, but the observed \emph{decline} within the first few months is unlikely to be causal.
It likely reflects a compositional effect specific to Portugal’s dual labor market, where a large share of new hires enter on temporary or subsidized traineeship contracts (\emph{estágios}) with below-market pay \citep{nunes2023failing}.
These workers disproportionately populate the lowest-tenure segment, temporarily pulling down average wages at the 5–10 month mark.
As such contracts expire or convert to permanent positions, the wage profile rises sharply thereafter.
From a theoretical standpoint, the same pattern may also reflect compensating differentials: workers accept lower initial pay in exchange for firm-provided training \citep{gregory2023firms} or future mobility opportunities \citep{del2024importance}.
Both mechanisms are consistent with a positive long-run slope of tenure returns once the early probationary phase is over.

\paragraph{Firm size, education, and qualification}
We conclude by revisiting the firm-size wage premium, a canonical finding in empirical labor economics. 
Figure~\ref{fig:pdp_size_edu} plots predicted log wages against firm size (log scale), conditional on worker qualification and education.
This suggests that the conventional firm-size wage premium largely reflects compositional differences and productivity differentials: larger firms employ more educated workers and generate higher revenue per worker, but once these attributes are accounted for, the TWICE model shows no residual size-related wage advantage.

Crucially, this result is not an artifact of collinearity between firm size and total revenue. 
Because size, revenue, and productivity (revenue per worker) are mechanically linked, care must be taken in defining the \emph{ceteris paribus} condition. 
When generating these plots, we hold revenue per worker (productivity) fixed at its sample median while varying firm size. 
Consequently, the flat profiles in Figure \ref{fig:pdp_size_edu} shows that firm size has no predictive power independent of productivity and worker composition: expanding the workforce without a corresponding increase in productivity yields no wage premium in our model.

\begin{figure}[t]
\centering
\includegraphics[width=0.9\textwidth]{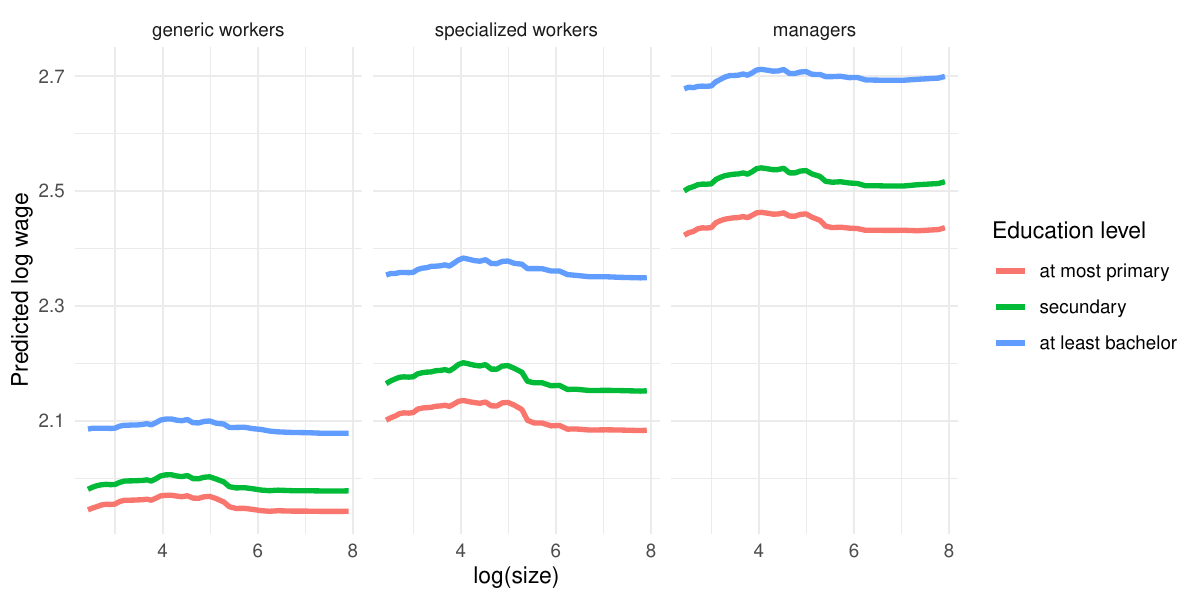}
\caption{PDP for firm size, education, and qualification}
\note[Note]{
Partial Dependence Plots showing the predicted log hourly wage as a function of firm size (log number of workers).
Panels are stratified by worker qualification. 
Lines represent education levels. 
The steepness of the curve indicates the return to firm-specific size, holding all other factors (including firm per-worker revenue) constant at their sample averages. 
The flat slopes indicate that, once productivity and worker composition are held constant, firm size itself has negligible predictive power for wages.
}
\label{fig:pdp_size_edu}
\end{figure}

Across all qualification groups, predicted wages remain essentially flat—or even decline slightly—as firm size increases. 
Education and qualification continue to strongly stratify wage levels, but the slope of the firm-size profile is close to zero throughout.
This suggests that the conventional firm-size wage premium largely reflects compositional differences: larger firms employ more educated and higher-qualified workers, but once these attributes and other firm observables (e.g., productivity, sector, and capital intensity) are accounted for, firm size itself carries little independent predictive power. 
For managers, the mild decline at the largest firm sizes may reflect pay compression in highly structured corporate environments, where hierarchical distance and centralized wage policies limit individual rent sharing \citep{card2018firms}. 
Thus, the TWICE model provides a more nuanced interpretation of the “firm-size premium”: it is not a direct causal effect of size but an emergent outcome of sorting and correlated firm characteristics already captured by the model’s flexible structure.

\subsection{Accumulated Local Effects}\label{ss:ale}

PDPs provide a useful summary of average model behavior, but they can be misleading when predictors are strongly correlated.
By construction, a PDP for a feature $X_j$ averages $\widehat f$ over the \emph{marginal} distribution of the remaining covariates $X_{-j}$, evaluating the model at combinations $(x_j, x_{-j})$ that may be rare or never observed in the data.
When features are correlated (e.g., age and tenure), this implies extrapolating the model off the empirical support, so the PDP need not coincide with the more natural target $\E{Y \mid X_j = x_j}$.

To address this limitation, we complement the PDPs with \emph{Accumulated Local Effects} (ALE) plots, which summarize the local influence of a feature on the prediction while respecting the joint distribution of the covariates.
For a continuous feature $X_j$, let $X = (X_j, X_{-j})$ and $\widehat f$ denote the fitted prediction function.
The first-order ALE function is defined as
\[ \mathrm{ALE}_j(x)
= \int_{x_j^{(0)}}^x 
\E[X_{-j} \mid X_j = z]{\frac{\partial \widehat f\of{X_j, X_{-j}}}{\partial x_j} \mid X_j = z}
\, dz
\]
where $x_j^{(0)}$ is the lower bound of the observed support of $X_j$.
Intuitively, $\mathrm{ALE}_j(x)$ accumulates the average \emph{local marginal effect} of $X_j$ on predicted wages at $X_j$ moves from $x_j^{(0)}$ to $x$, averaging over the conditional distribution of the remaining covariates. 
In practice, we approximate the derivative by finite differences and the conditional expectation by within-bin averages, so the resulting curve is defined only where the data are dense.

ALE plots thus quantify how small local changes in a feature shift predicted log-wages on average, accumulated along the empirical distribution of that variable.
Unlike PDPs, ALEs do not extrapolate beyond the data and remain robust under strong feature correlations—a common feature of employer–employee datasets where variables such as firm size, productivity, and capital intensity are highly collinear.

All ALEs are computed using the final TWICE LightGBM model described in Section~\ref{ss:implementation}, with the same cross-fitted structure used in the PDP analysis.
Together, the PDPs and ALEs provide complementary perspectives: the former highlight average global relationships, while the latter trace local, data-supported sensitivities of predicted wages to key covariates.

\paragraph{ALEs for tenure, productivity, and age}
Figure~\ref{fig:ale_main} reports Accumulated Local Effects (ALEs) for three covariates: tenure, revenue per worker, and age. 
ALEs are mean–centered partial effects, so each curve is interpreted as the deviation (in log points) of predicted wages from their average as the focal variable varies locally over its observed support.

\begin{figure}[t]
\subcaptionbox{Tenure (months, log scale)\label{fig:ale_tenure}}{\includegraphics[width=0.32\textwidth]{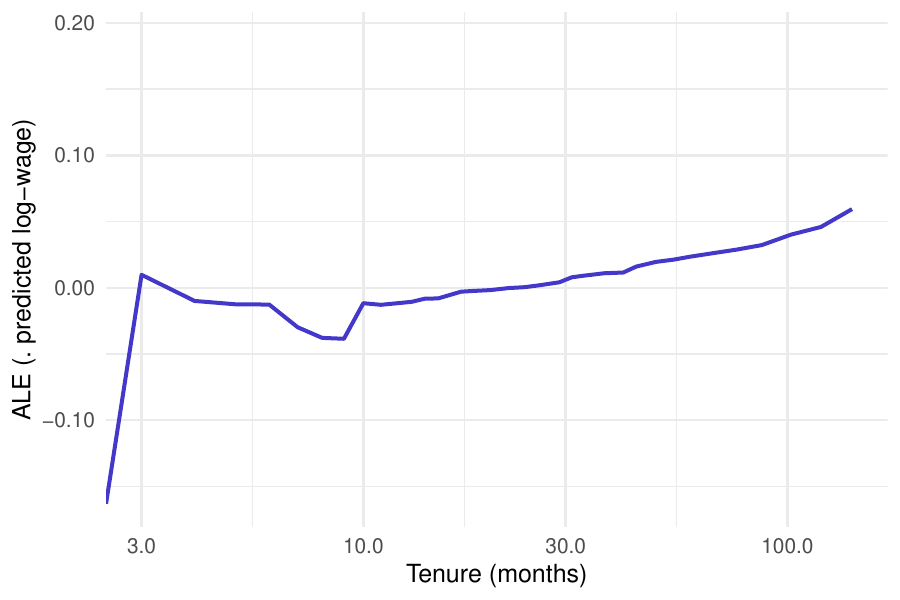}}\hfill
\subcaptionbox{Log(revenue per worker)\label{fig:ale_logrpw}}{\includegraphics[width=0.32\textwidth]{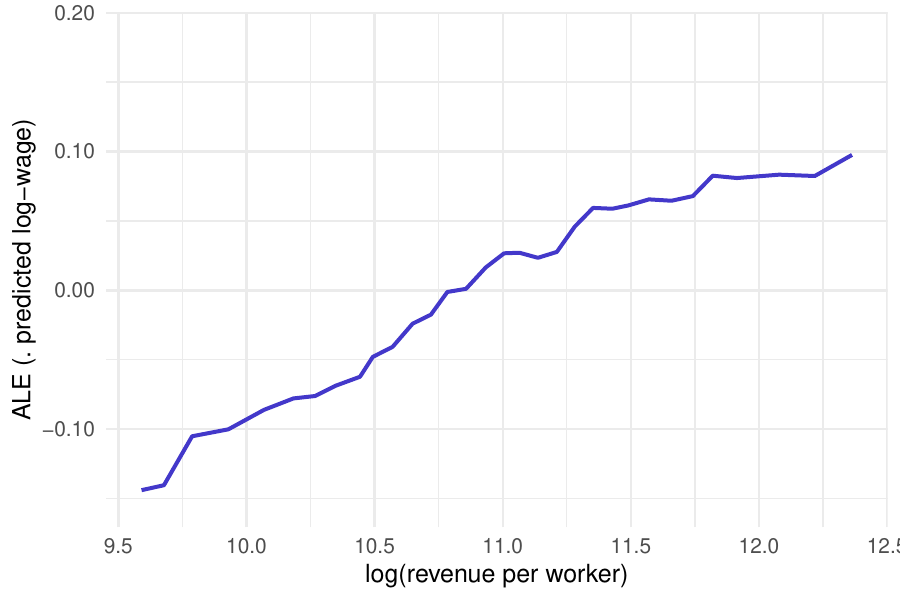}}\hfill
\subcaptionbox{Age (years)\label{fig:ale_age}}{\includegraphics[width=0.32\textwidth]{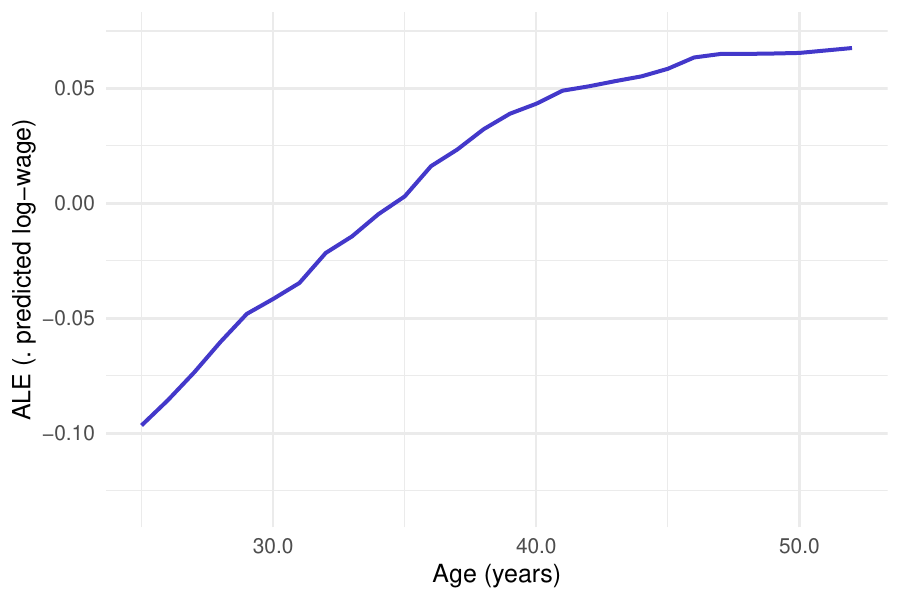}}
\caption{Accumulated Local Effects (ALEs) for tenure, productivity, and age}
\note[Note]{
Accumulated Local Effects (ALE) plots for three continuous covariates. 
Unlike PDPs, ALEs calculate marginal effects locally (using data within a neighborhood) to avoid evaluating the model on unrealistic combinations of correlated covariates (e.g., high tenure with young age). 
The y-axis represents the accumulated change in log wages relative to the average prediction.
}
\label{fig:ale_main}
\end{figure}

Panel~\ref{fig:ale_tenure} shows a modest dip at very short tenures (on the order of a few months) followed by a steady rise that becomes pronounced beyond roughly one year. 
The early trough is consistent with institutional features of the Portuguese labor market—e.g., the prevalence of fixed-term traineeships or probationary arrangements at below-standard pay—after which wages increase with firm-specific experience. 
By five to ten years of tenure the cumulative ALE is clearly positive, indicating meaningful firm-specific returns.

Panel~\ref{fig:ale_logrpw} displays a monotone, concave profile in log(revenue per worker): higher productivity per worker is associated with higher predicted wages, with the marginal effect tapering at the very top. 
This is consistent with rent-sharing or efficiency-wage mechanisms and provides a clean benchmark for firm “quality” that is less entangled with size \emph{per se}.

Panel~\ref{fig:ale_age} recovers a familiar concave age–earnings profile: 
the ALE rises steeply through the late 20s and 30s, flattens in the 40s, and levels off thereafter. The saturation at older ages is consistent with standard human-capital and internal-labor-market models.

Taken together, these ALEs complement the PDP evidence. 
The tenure and age profiles validate classic human-capital patterns in a way that is robust to feature correlation, while the strong positive gradient in revenue per worker clarifies that the negative slope observed in the firm-size PDPs is not a general property of “firm quality” but reflects conditioning choices and the joint distribution of size and productivity in the data.

\subsection{Variance decomposition}\label{ss:decomposition}

Section~\ref{ss:population_estimands} introduced the population-level variance decomposition implied by the TWICE framework.
TWICE partitions workers and firms using observable characteristics, estimates the corresponding cell means $\mu_{\ell k}$, and decomposes $\var{Y}$ into worker, firm, sorting, interaction, and residual components via additive projection.
The additive AKM model and the static BLM model deliver analogous decompositions, but differ in how worker and firm heterogeneity are represented. 
AKM relies on latent individual and firm fixed effects, while BLM estimates latent worker and firm classes and allows unrestricted dependence between them.

We now compare how the three frameworks allocate wage dispersion across these components.

\paragraph{Results}
Table~\ref{tab:joint_deco} reports the variance decomposition results for the AKM, BLM (static), and TWICE models side by side. 
All three refer to the same wage variable and sample, yielding a nearly identical total variance of approximately 0.17.\footnote{%
AKM variance is computed on the subsample that keeps only observations with both worker and firm FE estimates.
This explain the tiny discrepancies between AKM variance and the others.}

\begin{table}[t]
\caption{Variance decomposition of log wages under AKM, BLM (static), and TWICE}
\label{tab:joint_deco}
\centering
{
\renewcommand{\arraystretch}{1.2}
\begin{tabular*}{\textwidth}{p{3.9cm}@{\extracolsep\fill}cccccc}
\toprule
& \multicolumn{2}{c}{\textbf{AKM}} & \multicolumn{2}{c}{\textbf{BLM (static)}} & \multicolumn{2}{c}{\textbf{TWICE}} \\
\cmidrule(lr){2-3}\cmidrule(lr){4-5}\cmidrule(lr){6-7}
\textbf{Component} & \textbf{Var.} & \textbf{Share} & \textbf{Var.} & \textbf{Share} & \textbf{Var.} & \textbf{Share} \\
\midrule
Worker effect / group & 0.103 & 0.592 & 0.086 & 0.500 & 0.048 & 0.276 \\
Firm effect / group & 0.034 & 0.195 & 0.009 & 0.052 & 0.015 & 0.087 \\
Sorting $2\mathrm{Cov}(\alpha,\psi)$ & 0.013 & 0.072 & 0.034 & 0.195 & 0.020 & 0.116 \\
Interaction $\mathrm{Var}(\kappa_{gh})$ & --- & --- & --- & --- & 0.013 & 0.073 \\
Within-group variance & 0.022 & 0.124 & 0.044 & 0.254 & 0.077 & 0.448 \\
\addlinespace
\textit{Total variance} $\mathrm{Var}(Y)$ & 0.174 & 1.000 & 0.172 & 1.000 & 0.172 & 1.000 \\
\bottomrule
\end{tabular*}
}

\note[Note]{
Variance decomposition of log hourly wages estimated on the largest connected set. 
\emph{AKM}: standard two-way fixed-effects model with individual worker and firm effects.
\emph{BLM (static)}: latent-class model of \cite{bonhomme2019distributional} with $K=10$ firm classes and $L=6$ latent worker types estimated via EM.
\emph{TWICE}: observable-anchored partitions with $K=512$ firm classes and $L=512$ worker classes, selected by minimizing blocked out-of-sample prediction error (Table \ref{tab:grid_KL_loss}). 
Worker and firm effects are obtained via additive projection (weighted two-way ANOVA on cell means), ensuring orthogonality between marginal effects and the interaction term.
The interaction term $\var{\kappa_{gh}}$ captures non-additive complementarities unique to specific worker-firm cell combinations.
Within-group variance includes unobserved heterogeneity, idiosyncratic match effects, and measurement error.
}
\end{table}

In the standard AKM specification, worker fixed effects account for about 60\% of total variance, firm effects for 20\%, and sorting for 7\%. 
The static BLM model reallocates variance across components once correlated heterogeneity is allowed: the worker share declines to 50\%, the pure firm share to about 5\%, and sorting increases to roughly 20\%.

The TWICE decomposition yields notably different shares: worker groups account for 28\% of variance, firm groups for 8.7\%, sorting for 11.6\%, and interactions for 7.3\%.
The within-group residual is 45\%, larger than in AKM (12\%) or BLM (25\%).

\paragraph{Interpreting the differences}
The lower worker and firm shares in TWICE relative to AKM and BLM reflect a fundamental difference in what the groups capture.
AKM uses individual fixed effects—each worker is their own “group”—so worker heterogeneity mechanically absorbs all persistent individual-level variation.
BLM estimates latent worker types via maximum likelihood, optimizing the type assignments to explain wage variation.
Although BLM uses few types, these are endogenously chosen to maximize explanatory power.
TWICE groups workers by observable characteristics, which are interpretable but imperfect proxies for latent ability—two workers in the same TWICE cell may have very different true productivities.

TWICE's firm share (8.7\%) is comparable to BLM's (5.2\%), suggesting that observable firm characteristics capture a similar amount of systematic firm-level variation as BLM's latent classes.
Both are well below AKM's 19.5\%, which includes firm-specific noise that does not generalize out of sample.

\paragraph{The interaction term}
The TWICE interaction component—$\var{\kappa_{gh}}$, accounting for 7.3\% of variance—captures non-additive complementarities between observable worker and firm types.
This share is modest but indicates that some worker-firm combinations yield wages that deviate systematically from what additive worker and firm effects would predict.

The small interaction share is consistent with findings in \cite{bonhomme2019distributional}, who document “strong sorting and weak complementarities” in Swedish data.
In their reallocation exercise, randomly reassigning workers across firms changes mean log-wages by only 0.5\%, and adding interactions to their model improves $R^2$ by less than one percentage point.
Our results suggest a similar pattern holds for Portugal: workers sort substantially on observables, but conditional on this sorting, the wage function is approximately additive.

\paragraph{The within-group residual}
TWICE's within-group variance (44.8\%) is larger than in AKM or BLM because observable characteristics provide coarser descriptions of worker and firm heterogeneity than individual fixed effects or optimized latent types.
This residual includes: 
(i) unobserved worker ability within demographic cells, 
(ii) unobserved firm heterogeneity within balance-sheet cells, 
(iii) idiosyncratic match effects, and 
(iv) measurement error.

The 44.8\% within-cell share should not be interpreted as “unexplained” variance in the sense of model failure.
Rather, it reflects the deliberate trade-off in TWICE between interpretability (groups defined by observables) and explanatory power (variance captured).
The between-cell component represents the share of wage dispersion attributable to observable worker and firm characteristics and their interaction—a quantity directly relevant for policy analysis.
The fact that the cross-fitted model achieves high out-of-sample accuracy ($R^2 \approx 0.50$, Table~\ref{tab:oos}) confirms that observables explain a substantial share of systematic wage variation.

\paragraph{Contrasts across models}
Two contrasts emerge. 

First, relative to AKM, both BLM and TWICE attribute a smaller share of dispersion to pure worker heterogeneity: 60\% in AKM, 50\% in BLM, and 28\% in TWICE. 
For BLM, this reflects the reallocation of variance to sorting once correlated heterogeneity is modeled.
For TWICE, it additionally reflects that observable characteristics are imperfect proxies for latent ability.

Second, the role of sorting increases as the model allows for richer dependence between worker and firm heterogeneity. 
Sorting rises from 6\% in AKM to 20\% in BLM and to 12\% in TWICE.
The TWICE sorting component should be interpreted as assortative matching on observables: because worker and firm types are functions of $(Z,X)$, any systematic association between worker and firm characteristics that affects wages is reflected in $\cov{\alpha_g, \psi_h}$.
The lower sorting share in TWICE relative to BLM (12\% vs. 20\%) likely reflects that BLM's latent types capture unobserved dimensions of sorting that do not project onto our observable covariates.

\paragraph{What TWICE adds}
Relative to AKM and BLM, the main contribution of TWICE is not that it reverses their qualitative message—our results are broadly consistent with \cite{bonhomme2019distributional} and \cite{bonhomme2023much}: pure firm effects are modest once sorting is modeled flexibly, and complementarities are weak.
What TWICE adds is:

\begin{enumerate}
    \item \emph{Anchoring heterogeneity in observables.}
    Worker and firm premia are functions of rich, interpretable covariates rather than purely latent effects.
    This lets us link variance components, sorting, and interactions directly to balance-sheet variables and workforce composition, an information set that the existing literature largely treats as residual controls.
    
    \item \emph{A clean interaction component.}
    By projecting cell means onto additive worker and firm effects, TWICE isolates a well-defined interaction term $\kappa^*_{gh}$ that is orthogonal to both margins. 
    The interaction share (7.3\%) quantifies non-additive complementarities between observable worker and firm types—a quantity that AKM cannot estimate by construction. 
    Our finding of small but nonzero interactions is consistent with BLM's conclusion of “weak complementarities” in matched employer-employee data.
    
    \item \emph{Out-of-sample, policy-relevant structure.}
    Because the decomposition is built from a predictor that generalizes well to held-out firms, the components can be used to speak about new firms and workers with given observables (Section~\ref{ss:oos}), rather than only about units in the mobility graph.
    This is useful for counterfactual or policy exercises that shift the distribution of observables (e.g., productivity or education) rather than the latent effects.
\end{enumerate}

\medskip
Viewed this way, TWICE sharpens and reinterprets the existing variance-decomposition results: it confirms that sorting is quantitatively important, shows that a meaningful share of dispersion comes from observable interactions between worker and firm traits, and maps both patterns into rich administrative observables that previous approaches leave largely unused.

\paragraph{Robustness}
As a robustness check, we examine whether the construction of the firm partition depends on the choice of the firm-level target used in the first-stage tree. 
Appendix \ref{app:robustness_firm} reports results obtained when the firm tree is retrained using 
(i) median wages and 
(ii) firm-level averages of cross-fitted residuals from a worker-only model. 
The resulting variance decompositions are very similar to the baseline, and the predicted firm components exhibit a very high correlation with the baseline specification. 
This confirms that our main findings are not driven by the specific choice of target used to form the firm types.

\subsection{Concordance with AKM fixed effects}

A natural question is whether the observable-anchored TWICE partitions capture the same worker and firm heterogeneity that AKM identifies through individual fixed effects.
To assess this, we compute $\eta^2$: the share of cross-sectional variance in AKM fixed effects explained by TWICE class membership.
For workers, this asks: how much of the variation in individual AKM effects $\widehat{\alpha}_i$ is accounted for by knowing which TWICE cell a worker belongs to?
For firms, we weight by the number of worker-year observations, so that $\eta^2$ reflects the explanatory power for the typical worker rather than the typical firm.

\begin{figure}[t!]
\caption{Distribution of AKM effects across TWICE classes}
\label{fig:akm_effects_by_class}
\subcaptionbox{Distribution of AKM worker effects across TWICE worker classes\label{fig:akm_workerFE_by_class}}{\includegraphics[width=0.8\textwidth]{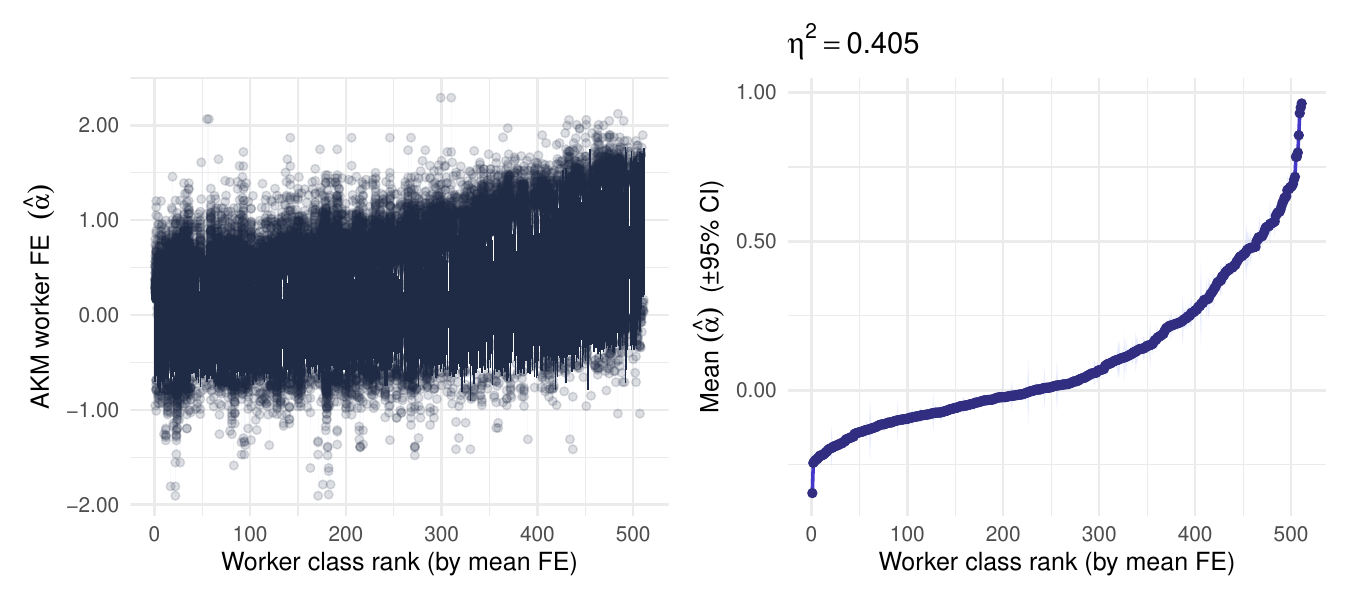}}\\
\subcaptionbox{Distribution of AKM firm effects across TWICE firm classes\label{fig:akm_firmFE_by_class}}{\includegraphics[width=0.8\textwidth]{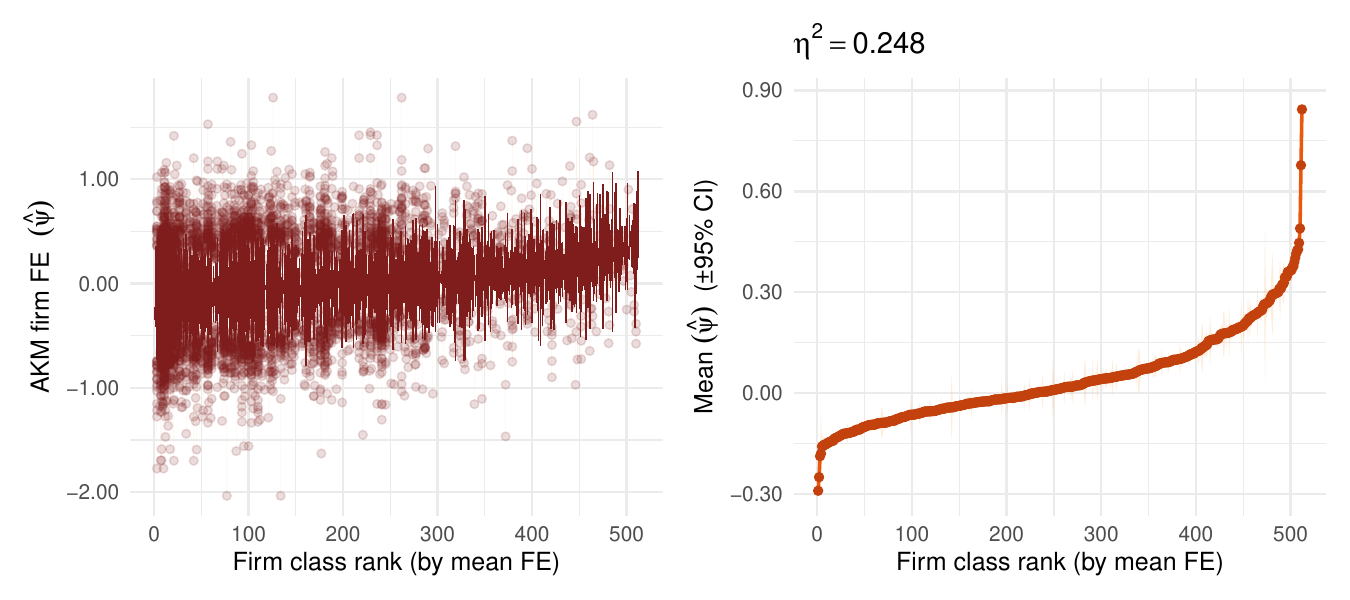}}
\note{
Concordance between observable-anchored TWICE partitions and latent AKM fixed effects.
Panel~A: Distribution of estimated AKM worker fixed effects ($\widehat{\alpha}_i$) within each TWICE worker class, ordered by class mean.
Panel~B: Distribution of estimated AKM firm fixed effects ($\widehat{\psi}_j$) within each TWICE firm class, ordered by class mean.
The red line connects class means. 
$\eta^2$ reports the share of variance in AKM fixed effects explained by TWICE class membership.
For workers, each individual receives equal weight ($\eta^2 = 0.40$).
For firms, each firm is weighted by its number of worker-year observations ($\eta^2 \simeq 0.25$), so the statistic reflects explanatory power for the typical worker.
}
\end{figure}

Figures~\ref{fig:akm_workerFE_by_class} and~\ref{fig:akm_firmFE_by_class} display the results. 
TWICE worker classes explain 40\% of the variance in AKM worker effects, indicating that observable characteristics—age, tenure, education, and occupation—capture a substantial share of persistent worker heterogeneity.
The remaining 60\% reflects unobserved ability differences within demographic cells: two workers with identical observables may have very different productivities.

On the firm side, TWICE classes explain around 25\% of the observation-weighted variance in AKM firm effects.
While lower than for workers, this concordance must be interpreted in light of the substantial sampling noise in AKM firm estimates.
In sparse mobility networks, a large fraction of the variance in $\widehat{\psi}_j$ reflects estimation error rather than true firm wage premia \citep{kline2020leave, bonhomme2023much}.
The 25\% captured by TWICE likely understates the share of \emph{true} firm heterogeneity explained by observables, since the denominator includes noise that no observable could predict.

This interpretation is supported by TWICE's out-of-sample performance.
Despite explaining only 25\% of the variation in $\widehat{\psi}_j$, TWICE achieves substantially higher predictive accuracy for wages than flexible OLS benchmarks (Table~\ref{tab:oos}).
The variation in AKM firm effects that TWICE misses appears to consist largely of sampling noise and idiosyncratic match effects that do not generalize, whereas TWICE captures the systematic component of firm wage policies anchored in balance-sheet observables.

Overall, the comparison confirms that TWICE partitions recover the predictable, observable-driven component of worker and firm heterogeneity.
The partitions provide interpretable groups that align with the persistent structure in wages without relying on mobility-identified latent effects.

\section{Conclusions}\label{s:conclusions}

This paper develops TWICE, a framework that integrates machine learning with the two-way decomposition tradition in matched employer--employee data. 
Rather than relying on latent worker and firm effects identified through mobility, TWICE models the conditional expectation of wages given rich observable characteristics, constructs observable-anchored worker and firm partitions, and summarizes wage dispersion through a decomposition into worker components, firm components, assortative matching, and non-additive interactions.

Applied to Portuguese administrative data, TWICE delivers high out-of-sample predictive accuracy, indicating that the estimated conditional wage function captures stable features of the wage structure rather than overfitting noise. 
The resulting variance decomposition reveals that assortative matching on observables accounts for 10\% of wage variance—compared to 6\% in standard AKM on the same data—while non-additive worker--firm interactions contribute an additional 5\%, a component that additive AKM models cannot capture by construction. 
These findings complement recent evidence showing that conventional fixed-effects estimators understate sorting and impose overly restrictive additive structure.

A key advantage of TWICE is its interpretability. 
Observable-anchored partitions allow us to characterize worker and firm heterogeneity in economically meaningful terms, and Partial Dependence and Accumulated Local Effects plots provide transparent diagnostics of how wages vary with specific attributes.
Because partitions are defined by observables rather than sample identifiers, the framework extends naturally to new workers and firms—enabling out-of-sample prediction and policy-relevant counterfactuals that latent-class or fixed-effect methods cannot readily support.

TWICE trades the ability to capture purely idiosyncratic unobserved heterogeneity for robustness to limited mobility and sampling noise. 
The framework isolates the component of firm pay policies predictable from observables, relegating idiosyncratic variation to the residual. 
Although we do not provide formal inference for the variance shares, extending recent results on inference under multiway clustering to this setting is a promising direction for future work.

Overall, TWICE provides a flexible, transparent, and out-of-sample-valid approach to analyzing wage dispersion in matched data. 
The general strategy—using supervised learning to form interpretable partitions and combining them with cross-fitted prediction—may prove useful in other settings with two-sided heterogeneity, such as credit markets, health care, or education. 
Integrating TWICE with causal methods to study policy counterfactuals or structural complementarities is a natural avenue for future research.

\bibliography{\bib}
\appendix


\section{A brief review of the existing methods}\label{a:AKM_BLM}
The analysis of wage formation with two-way heterogeneity has evolved through different approaches.
Here, we focus on the standard \cite{abowd1999high} (AKM) model and the recent refinement by \cite{bonhomme2019distributional} (BLM).

\subsection[The AKM model]{The \citetalias{abowd1999high} model}
\cite{abowd1999high} introduced a fixed effects wage model that can be expressed as:
\[ Y_{it} = \alpha_i + \psi_{j(i,t)} + X_{it}'\beta + \varepsilon_{it} \]
where $Y_{it}$ denotes worker $i$'s log-wage in year $t$, $j(i,t) \in \{1, \dots, J \}$ identifies the firm employing worker $i$ at $t$, $X_{it}$ includes control variables—often an age polynomial and time effects.\footnote{Following \citet{card2018firms}, implementations typically use an age polynomial normalized to be flat at a reference age and exclude the linear term to avoid collinearity with worker and year fixed effects.}

Identification relies on two assumptions. 
First, \emph{exogenous mobility}: conditional on observables, worker moves are orthogonal to the idiosyncratic error $\varepsilon_{it}$.
This implies that, absent firm effect differences, job changes would not systematically affect wages. 
Second, \emph{additivity}: worker and firm components enter additively, ruling out complementarities between $\alpha_i$ and $\psi_{j(i,t)}$.
\citetalias{abowd1999high} effects are identified only on the largest \emph{connected set} of the employer–employee mobility graph \citep{abowd2002computing}.
Intuitively, when two firm groups are not linked by flows, their relative levels are not pinned down.

\subsection{The BLM model}
\citet{bonhomme2019distributional} (BLM) address the dimensionality and additivity of \citetalias{abowd1999high} by grouping firms into a finite number of latent classes and modeling worker heterogeneity using a finite mixture of unobserved types. 
This allows for rich, unrestricted interactions between worker and firm heterogeneity.

BLM proposes a two-stage approach to overcome the computational challenges of simultaneous estimation. 
First, they use k-means clustering on firm-level empirical earnings distributions to assign each firm $j$ to one of $K$ classes. 
Second, conditional on these class assignments, they estimate a finite mixture model where the distribution of earnings depends on the interaction between a worker's latent type and their employer's class. 
For example, in their application, the mean and variance of log-earnings are allowed to differ for each worker type-firm class pair.

They estimate this second stage using an Expectation-Maximization (EM) algorithm to handle the latent nature of the worker types. 
This approach parsimoniously models worker heterogeneity while allowing its effect on earnings to be moderated by the firm class, capturing potential complementarities.

The BLM framework overcomes two key limitations of \citetalias{abowd1999high}: it relaxes the additivity assumption to allow for unrestricted interactions and, by grouping firms, it increases the density of the mobility network, which helps address the “limited mobility bias” documented by \cite{andrews2008high} and \cite{bonhomme2023much}.

\subsection{Limitations and motivation for TWICE}
The AKM and BLM frameworks have established foundational insights into wage dispersion. 
AKM identified firm heterogeneity as a first-order component, while BLM refined this by showing that sorting is paramount and that pure firm effects, net of workforce composition, are modest. 
Both show that parsimonious representations are crucial when mobility is sparse.

For our goals, three limitations remain. 
First, both frameworks treat the main objects as latent and typically relegate observables to residual controls: they are not designed to map balance-sheet and workforce characteristics directly into wage premia or to generalize those premia to new firms and workers. 
Second, forming firm classes from earnings distributions can blur wage-setting policies with workforce composition, making it harder to separate the two.
Third, portability outside the connected set (AKM) or beyond the training firms (BLM classes) is limited. 
These considerations motivate a complementary, \emph{observable-anchored} approach that retains the two-way structure and sorting/decomposition logic while emphasizing prediction, interpretability, and out-of-sample use.

\section{Theoretical framework and identification}\label{a:theory}

This appendix provides the formal statistical foundation for TWICE. 
It makes explicit the data generating process, states the assumptions on observables and on the learner, and links the variance decomposition to the conditional expectation function (CEF) $m_0$ introduced in Section~\ref{ss:population_estimands}.

\subsection{Data generating process}

Let $Y_{it}$ denote the log wage of worker $i$ at time $t$, employed at firm $j=J(i,t)$. 
Let $X_j \in \mathcal{X} \subseteq \mathbb{R}^{d_x}$ and $Z_{it} \in \mathcal{Z} \subseteq \mathbb{R}^{d_z}$ denote firm and worker observables, respectively ($\mathcal{X}$ and $\mathcal{Y}$ denote the support of either).
We assume
\begin{equation*}
    Y_{it} = m_0(X_j, Z_{it}) + u_{it},
\end{equation*}
where $m_0 : \mathcal X \times \mathcal Z \to \mathbb R$ is a structural wage function and $u_{it}$ is an idiosyncratic error.

Identification of $m_0$ relies on the following:
\begin{assumption}[Sufficiency of observables]\label{ass:exogeneity}
The observables $(X_j,Z_{it})$ contain all systematically wage–relevant information, in the sense that
\[
    \E{u_{it} \mid X_j, Z_{it}} = 0.
\]
\end{assumption}

Let $P_{X,Z}$ denote the joint distribution of $(X_j,Z_{it})$.

\begin{proposition}[CEF representation]\label{prop:cef}
Under the data generating process above and Assumption~\ref{ass:exogeneity}, 
\[
    \E{Y_{it} \mid X_j = x, Z_{it} = z} = m_0(x,z)
\]
for $P_{X,Z}$–almost every $(x,z)$.
\end{proposition}

\begin{proof}
By the law of iterated expectations,
\[
\E{Y_{it} \mid X_j,Z_{it}} 
= \E{ m_0(X_j,Z_{it}) + u_{it} \mid X_j,Z_{it}}
= m_0(X_j,Z_{it}) + \E{u_{it} \mid X_j,Z_{it}}
\]
Assumption~\ref{ass:exogeneity} implies $\E[u_{it} \mid X_j,Z_{it}] = 0$, hence
\(
\E{Y_{it} \mid X_j,Z_{it}} = m_0(X_j,Z_{it})
\)
which yields the stated equality for $P_{X,Z}$–almost every $(x,z)$.
\end{proof}

\subsection{Population partitions and approximation}

Section~\ref{ss:population_estimands} defined the worker and firm partition functions 
$g:\mathcal Z\to\{1,\dots,L\}$ and $h:\mathcal X\to\{1,\dots,K\}$, 
the $L^2$–optimal population partitions $(g^*,h^*)$, and the associated cell means
\[
    \mu_{\ell k} \equiv \E{Y_{it} \mid g^*(Z_{it})=\ell,\; h^*(X_j)=k},
    \qquad
    \mu \equiv \E{Y_{it}}
\]
For notational convenience, write $c_{it} \equiv \mu_{g^*(Z_{it}), h^*(X_j)}$ for the cell mean associated with observation $(i,t)$.
The residual in the decomposition is $\xi_{it} \equiv Y_{it} - c_{it}$.

Substituting $Y_{it} = m_0(X_j,Z_{it}) + u_{it}$ yields
\begin{equation*}
    \xi_{it}
    = \underbrace{ \big( m_0(X_j, Z_{it}) - c_{it} \big) }_{\text{approximation error}}
      + \underbrace{ u_{it} }_{\text{stochastic error}}
\end{equation*}
By the definition of $\mu_{\ell k}$ as conditional means,
\[
    \mu_{\ell k}
    = \E{Y_{it} \mid g^*(Z_{it})=\ell,\; h^*(X_j)=k}
    = \E{m_0(X_j,Z_{it}) \mid g^*(Z_{it})=\ell,\; h^*(X_j)=k}
\]
so
\[
    \E{ m_0(X_j,Z_{it}) - c_{it} \mid g^*(Z_{it}) = \ell,  h^*(X_j)=k } = 0
\]
Hence the approximation error is mean–zero within each cell and can be pooled with $u_{it}$ in the residual variance component $\var(\xi_{it})$.
Intuitively, $\xi_{it}$ collects pure noise plus within–cell heterogeneity that cannot be captured by a finite number of worker and firm groups.

\subsection{Approximation and estimation}\label{app:approx_est}

In practice, we do not observe $m_0$ and cannot directly work with $(g^*,h^*)$.
Instead, TWICE uses a sequence of tree–based predictors to approximate $m_0$ and to construct estimated partitions.
In this section, we formally characterize the properties of this approximation. For notational simplicity, we drop observation indices and write $(X,Y,Z)$ for $(X_j,Y_{it},Z_{it})$.

Let $\{\mathcal F_N\}_{N\ge 1}$ denote a (possibly growing) sequence of function classes from $\mathcal X \times \mathcal Z$ to $\mathbb R$ induced by the tuning restrictions of the tree--ensemble learner (depth, minimum leaf size, number of boosting iterations, early stopping, etc.). Let $L^2(P_{X,Z})$ be the (real) Hilbert space of measurable functions defined on the underlying sample space such that the squared function is integrable with respect to the measure 
$P_{X,Z}$. Define the best-in-class approximation to $m_0$ as any element
\[
    f_N^* \in \arg\min_{f\in\mathcal F_N} \E{ \bp{ m_0(X,Z) - f(X,Z) }^2 }.
\]
We do not require $f^*_N=m_0$. If, as $N\to\infty$, the approximation error 
\[ \inf_{f\in\mathcal F_N} \E{\bp{ m_0(X,Z)-f(X,Z)}^2 } \] 
tends to zero (e.g. if $\cup_N\mathcal F_N$ is dense in $L^2(P_{X,Z})$), then $f^*_N$ approximates $m_0$ arbitrarily well in $L^2(P_{X,Z})$.
Let $\widehat f_N$ be the predictor produced by the gradient–boosted tree algorithm on a sample of size $N$, with tuning parameters chosen as described in Section~\ref{ss:estimation}.
We posit the following assumption(s).

\begin{assumption}[Regularity for tree-based learners]\label{ass:tree_reg}
    The following conditions hold:
    \begin{enumerate}
        \item The support of $(X,Z)$ is compact and the CEF $m_0(x,z)$ is bounded and Lipschitz–continuous.
        \item The function class $\mathcal F_N$ consists of additive tree ensembles (finite sums of regression trees), with depth allowed to grow with $N$.
        \item Define the population and sample mean squared prediction errors
        \[
        Q(f) := \E{\bp{Y-f(X,Z)}^2}, \qquad Q_N(f) := \frac 1 N \sum_{n=1}^N \bp{Y_n-f(X_n,Z_n)}^2.
        \]
        The tuning/estimation procedure returns $\widehat f_N\in\mathcal F_N$ such that:
        (i) it is \[\sup_{f\in\mathcal F_N}\big|Q_N(f)-Q(f)\big|\xrightarrow{p}0;\]
        (ii) there exists $\varepsilon_N\downarrow 0$ with
        \[
        Q_N(\widehat f_N)\le \inf_{f\in\mathcal F_N}Q_N(f)+\varepsilon_N
        \quad\text{with probability tending to one.}
        \]
        \item Moments are uniformly bounded: $\E{Y^2} < \infty$ and $\sup_{f \in \mathcal F_N} \E{f(X,Z)^2} < \infty$ for all $N$.
        \item For each $N$, the set $\mathcal F_N$ is nonempty, closed, and convex in $L^2(P_{X,Z})$, so that the minimum defining $f_N^*$ is attained and $f_N^*$ is the $L^2(P_{X,Z})$--projection of $m_0$ onto $\mathcal F_N$.
    \end{enumerate}
\end{assumption}

These are standard regularity conditions ensuring that the problem is well-behaved. Existing results for tree ensembles imply that, with suitable tuning, the expected mean squared prediction error of $\widehat f_N$ can be made arbitrarily close to the best achievable within $\mathcal F_N$; see, for example, \citet{atheywager2018} and \citet{chiang2022multiway}. Proposition \ref{prop:consistency} restate these results in our setting.

\begin{proposition}[Prediction--error consistency]\label{prop:consistency}
Suppose Assumption \ref{ass:exogeneity} holds and Assumption \ref{ass:tree_reg} holds for the sequence of classes $\{\mathcal F_N\}$.
Let $f^*_N \in \arg\min_{f\in\mathcal F_N} \E{m_0(X,Z)-f(X,Z)^2}$.
Then
\[
\E{\bp{\widehat f_N(X,Z) - f_N^*(X,Z)}^2} \to 0 \quad\text{as } N\to\infty.
\]
If, moreover, $\inf_{f\in\mathcal F_N}\E{\bp{m_0(X,Z)-f(X,Z)}^2}\to 0$, then
\[
\E{\bp{\widehat f_N(X,Z)-m_0(X,Z)}^2}\to 0.
\]
\end{proposition}

\begin{proof}
Fix $N$. For any $f\in\mathcal F_N$,
\[
Q(\widehat f_N)-Q(f) = \{Q(\widehat f_N)-Q_N(\widehat f_N)\} + \{Q_N(\widehat f_N)-Q_N(f)\}
+\{Q_N(f)-Q(f)\}.
\]
Taking $f=f_N^*$ and using Assumption \ref{ass:tree_reg}(c)(ii),
\[
Q(\widehat f_N)-Q(f_N^*) \le \{Q(\widehat f_N)-Q_N(\widehat f_N)\} +\{Q_N(f_N^*)-Q(f_N^*)\}
+\varepsilon_N.
\]
Hence,
\[
Q(\widehat f_N)-Q(f_N^*) \le 2\sup_{g\in\mathcal F_N}\big|Q_N(g)-Q(g)\big|+\varepsilon_N.
\]
Assumption \ref{ass:tree_reg}(c)(i) and $\varepsilon_N\to 0$ imply
$Q(\widehat f_N)-Q(f_N^*)\xrightarrow{p}0$.

Under Assumption \ref{ass:exogeneity}, $Y=m_0(X,Z) + u$ with $\E{\bp{u\mid X,Z}}=0$, so for any $f$,
\[
Q(f)= \E{\bp{m_0(X,Z)-f(X,Z)}^2} + \E{u^2},
\]
and therefore
\[
Q(\widehat f_N) - Q(f^*_N) = \E{\bp{m_0(X,Z)-\widehat f_N(X,Z)}^2} - \E{\bp{m_0(X,Z)-f^*_N(X,Z)}^2}.
\]
If $\mathcal F_N$ is closed and convex in $L^2(P_{X,Z})$, then $f^*_N$ is the $L^2$--projection of $m_0$ onto $\mathcal F_N$. Hence, by the properties of projections on Hilbert spaces:
\[
\E{\bp{m_0-\widehat f_N}^2} \ge \E{\bp{m_0 - f^*_N}^2} + \E{\bp{\widehat f_N - f^*_N}^2},
\]
so that:
\[
\E{\bp{\widehat f_N - f^*_N}^2} \le Q_N(\widehat f_N)-Q_N(f^*_N) \xrightarrow{p}0.
\]
Since $\E{\bp{\widehat f_N-f^*_N)^2}}\ge 0$, Assumption \ref{ass:tree_reg}(d) implies that
$\E{\bp{\widehat f_N-f^*_N)^2}}_N$ is uniformly integrable; hence
$\E{\bp{\widehat f_N-f^*_N)^2}}\to 0$.

Finally, by the triangular inequality:
\[
\|\widehat f_N-m_0\|_{L^2}
\le \|\widehat f_N-f^*_N\|_{L^2}+\|f^*_N - m_0\|_{L^2},
\]
and if the approximation error $\|f^*_N-m_0\|_{L^2}\to 0$, then $\|\widehat f_N-m_0\|_{L^2}\to 0$.
\end{proof}

\subsection{Cross--fitting and orthogonality}

Variance components in TWICE are computed using predictions from cross--fitted versions of $\widehat f_N$.
Partition worker identifiers into disjoint blocks $\{\mathcal I_a\}_{a=1}^B$ and firm identifiers into $\{\mathcal J_b\}_{b=1}^B$.
For each $(a,b)$, let
\[
S_{ab}=\{(i,t): i\in \mathcal I_a,\; J(i,t)\in \mathcal J_b\}
\]
be the holdout cell. 
Define the corresponding training index set
\[
S^{\mathrm{tr}}_{ab}:=\{(i,t): i\notin \mathcal I_a,\; J(i,t)\notin \mathcal J_b\}.
\]
Estimate $\widehat f^{(-ab)}$ on $S^{\mathrm{tr}}_{ab}$ and generate predictions
$\widehat f^{(-ab)}(X_{J(i,t)},Z_{it})$ only for $(i,t)\in S_{ab}$.

We (somewhat informally) assume a standard multiway--clustering structure:

\begin{assumption}[Dependence structure]\label{ass:dependence}
Errors $\{u_{it}\}$ may be arbitrarily dependent within worker clusters and within firm clusters, but are weakly dependent across distinct workers and distinct firms (so that standard asymptotic results apply).
\end{assumption}

Cross-fitting ensures that $\widehat f^{(-ab)}$ is trained on a sample that excludes the evaluation observation $(i,t)\in S_{ab}$, so that the only remaining source of correlation with $u_{it}$ is approximation/estimation error, which vanishes per the following results.

\begin{proposition}[Cross--fitted orthogonality]\label{lem:orthogonality}
Suppose Assumptions~\ref{ass:exogeneity}, \ref{ass:tree_reg}, and \ref{ass:dependence} hold, and that the prediction--error consistency property in Proposition~\ref{prop:consistency} holds for each fold estimator $\widehat f^{(-ab)}$ trained on $S^{\text{tr}}_{ab}$.
Then, for each $(a,b)$ and each $(i,t)\in S_{ab}$ (with $j=J(i,t)$),
\[
    \E{ \widehat f^{(-ab)}(X_j, Z_{it}) u_{it} } \longrightarrow 0
    \qquad \text{as } N\to\infty.
\]
\end{proposition}

\begin{proof}
Fix $(a,b)$ and an observation $(i,t)\in S_{ab}$ with $j=J(i,t)$.
Write $f^*_N$ for the best-in-class function in Proposition~\ref{prop:consistency}, so that $f^*_N(X_j,Z_{it})$ is measurable with respect to $(X_j,Z_{it})$.

Decompose
\[
\E{\widehat f^{(-ab)}(X_j,Z_{it}) u_{it}}
=
\E{(\widehat f^{(-ab)}(X_j,Z_{it})-f^*_N(X_j,Z_{it}))\,u_{it}}
+
\E{f^*_N(X_j,Z_{it}) u_{it}}.
\]
The second term on the right-hand side equals zero by Assumption~\ref{ass:exogeneity} since, by the Law of Iterated Expectations:
\[
\E{f^*_N(X_j,Z_{it}) u_{it}}
=
\E{f^*_N(X_j,Z_{it}) \E{u_{it}\mid X_j,Z_{it}}}=0.
\]
Regarding the first term, the Cauchy--Schwarz inequality gives:
\begin{multline*}
\left| \E{ \bp{ \widehat f^{(-ab)}(X_j,Z_{it}) - f^*_N(X_j,Z_{it})} u_{it}} \right|\le\\ \le
\bp{\E{ \bp{\widehat f^{(-ab)}(X_j,Z_{it})-f^*_N(X_j,Z_{it})}^2}}^{1/2}
\cdot
\bp{ \E{u_{it}^2}}^{1/2}.
\end{multline*}
By Proposition~\ref{prop:consistency} applied to the fold estimator $\widehat f^{(-ab)}$ and the moment bound in Assumption~\ref{ass:tree_reg}(d), the first factor converges to zero. Because $\E{u_{it}^2}<\infty$, the full expression also converges to zero. This completes the proof.
\end{proof}

\subsection{Decomposition and sorting}
Given the population partitions $(g^*,h^*)$ and the cell means $\mu_{\ell k}$, the worker and firm components are defined via additive projection. 
Let $\pi_{\ell k}$ denote the population share of cell $(\ell, k)$, with marginals $\pi_\ell = \sum_k \pi_{\ell k}$ and $\pi_k = \sum_\ell \pi_{\ell k}$.
The projected effects $(\alpha^*, \psi^*)$ solve
\[
\min_{\alpha, \psi} \sum_{\ell k} \pi_{\ell k} \bp{\mu_{\ell k} - \mu -\alpha_\ell - \psi_k}^2
\]
subject to $\sum_\ell \pi_\ell \alpha_\ell = 0$ and $\sum_k \pi_k \psi_k = 0$.
This is a weighted two-way ANOVA decomposition of the cell means.
The first-order conditions yield:
\begin{align*}
    \alpha_\ell^* &= \sum_k \bp{\pi_{\ell k} / \pi_\ell} \bp{\mu_{\ell k} - \mu - \psi^*_k} \\
    \psi_k^* &= \sum_\ell \bp{\pi_{\ell k} / \pi_k} \bp{\mu_{\ell k} - \mu - \alpha^*_\ell}
\end{align*}
These form a system that can be solved iteratively or via direct matrix methods. 

The key property is that the interaction term $\kappa_{\ell k}^* = \mu_{\ell k} - \mu - \alpha_\ell^* - \psi_k^*$ satisfies
\[
\sum_{\ell, k} \pi_{\ell k} \alpha_\ell^* \kappa_{\ell k}^* = 0 
\qquad
\sum_{\ell, k} \pi_{\ell k} \psi_k^* \kappa_{\ell k}^* = 0 
\]
Hence $\cov{\alpha^*, \kappa^*} = \cov{\psi^*, \kappa^*} = 0$ at the observation level.

When $\kappa_{\ell k}=0$ for all $(\ell,k)$, the systematic component of wages is additive in worker and firm types; nonzero values of $\kappa_{\ell k}$ capture complementarities or mismatch at specific worker–firm combinations.

Moreover, when cells are balanced ($\pi_{\ell k} = \pi_\ell \pi_k$ for all $\ell, k$), the projection coincides with simple conditional means: $\alpha_\ell^* = \E{\mu_{\ell k} \mid \ell} - \mu$. 
Under sorting, where high-type workers concentrate in high-type firms ($\pi_{\ell k}  \neq \pi_\ell \pi_k$), the projection and conditional means differ. 
The projection isolates the “pure” worker effect from the firm types that workers of type $\ell$ endogenously match with.

The sorting component in the variance decomposition is
\[
    \text{Sorting} = 2 \cov{\alpha_{g^*(Z)}, \psi_{h^*(X)}}
    = 2 \sum_{\ell,k} \pi_{\ell k} \alpha_\ell \psi_k
\]
which measures \emph{assortative matching on observables}: the extent to which high–premium worker types match with high–premium firm types.

\subsection{Relationship with AKM}

The standard AKM model specifies
\[
    Y_{it} = \theta_i + \Psi_{J(i,t)} + r_{it}
\]
where $\theta_i$ and $\Psi_j$ are latent worker and firm effects and $r_{it}$ is an idiosyncratic error.
For comparison with the TWICE decomposition, it is useful to write the latent firm effect as
\[
    \Psi_j = \psi_{h^*(X_j)} + \nu_j
\]
where $\psi_{h^*(X_j)}$ is the firm component implied by the TWICE partition and $\nu_j$ is a residual term combining unobserved firm heterogeneity and estimation noise.

In this representation, $\psi_{h^*(X_j)}$ can be interpreted as the part of the AKM firm premium that is systematically associated with observable firm characteristics $X_j$, while $\nu_j$ captures features of $\Psi_j$ that are orthogonal to the TWICE partition (together with sampling noise in finite samples).
In sparse mobility networks, the variance of the sampling noise in $\widehat\Psi_j$ is large, implying $\var(\widehat \Psi) > \var(\Psi)$.
TWICE regularizes $\Psi_j$ by restricting attention to the component that is predictable from observables, effectively discarding much of the sampling noise as well as any unobserved firm heterogeneity that does not load on $X_j$.

This trade–off underlies the interpretation of TWICE firm effects as \emph{observable} wage premia and of TWICE sorting as assortative matching on observables, as discussed in Section~\ref{ss:identification}.

\section{Data and sample construction}\label{app:data}

This appendix provides details on sample selection, variable definitions, and descriptive statistics.

\subsection{Sample selection}

Starting from the matched QP-SCIE data for 2012--2019, we apply the following filters:
\begin{itemize}[nosep]
\item \emph{Firms:} Private sector; at least 5 employees; non-missing EBITDA.
\item \emph{Workers:} Full-time; aged 20-65; tenure $\geq$ 2 months; positive annual earnings.
\item \emph{Wages:} Log hourly wage, constructed as $\log(\text{annual earnings} / \text{annual hours})$.
\item \emph{Connectivity:} Restricted to the largest connected set of firms linked by worker mobility, ensuring comparability with AKM.
\end{itemize}
Table~\ref{tab:connectivity} reports connectivity statistics.
The largest connected set retains 99.6\% of workers and 97.6\% of firms, reflecting a well-integrated labor market with substantial worker mobility across employers.

\subsection{Descriptive statistics}

This subsection provides descriptive statistics for the matched employer--employee panel used in the empirical analysis. 
Unless otherwise noted, all moments are computed on the estimation sample restricted to the largest connected set of the worker--firm mobility graph over 2012--2019. 
\autoref{tab:stat_panel} summarizes annual cross-sections; \autoref{tab:stat_industry_19} reports the industry composition in 2019; \autoref{tab:connectivity} describes mobility and connectivity, which are key for AKM-style identification. \autoref{tab:firm_size_distrib} and \autoref{tab:tenure_age} document the distributions of firm size (number of employees observed in the sample), worker age, and firm tenure; hourly wages are expressed in logs when indicated.

\begin{table}[h!]
\caption{Panel summary by year}
\label{tab:stat_panel}
\centering
\begin{tabular*}{.9\textwidth}[]{p{2.5cm}@{\extracolsep\fill}ccccc}
\toprule
Year & \# Firms & \# Workers & \makecell{Avg. \\ log-wage} & \makecell{Sh. \\ graduates} & \makecell{Sh. \\ managers} \\
\midrule
2012 & 43,162 & 349,195 & 1.97 & 0.17 & 0.09 \\
2013 & 42,054 & 351,095 & 1.99 & 0.18 & 0.09 \\
2014 & 43,697 & 376,585 & 2.00 & 0.19 & 0.09 \\
2015 & 46,131 & 406,379 & 2.00 & 0.20 & 0.09 \\
2016 & 49,171 & 439,757 & 2.01 & 0.20 & 0.09 \\
2017 & 52,123 & 474,678 & 2.03 & 0.21 & 0.09 \\
2018 & 54,525 & 496,647 & 2.06 & 0.21 & 0.10 \\
2019 & 55,330 & 507,366 & 2.10 & 0.22 & 0.10 \\
\bottomrule
\end{tabular*}

\note[Note]{
Annual cross-sections from the matched employer-employee panel (Quadros de Pessoal), restricted to the largest connected set.
Log-wage is the natural logarithm of hourly wages.
Graduates: share with tertiary education (ISCED 5+).
Managers: share in managerial occupations.
The steady growth in employment and graduate share reflects Portugal's post-crisis recovery and workforce upskilling.
}
\end{table}

\begin{table}[h!]
\caption{Industry composition (2019)}
\label{tab:stat_industry_19}
\centering
\begin{tabular*}{\textwidth}[]{p{5cm}@{\extracolsep\fill}cccccc}
\toprule
Industry & \# Firms & \# Workers & \makecell{Avg.\\log-wage} & \makecell{Avg.\\firm size} & \makecell{Sh.\\graduates} & \makecell{Sh.\\managers} \\
\midrule
wholesale, retail, accommodation & 20,515 & 146,900 & 2.07 & 7.16 & 0.22 & 0.09 \\
manufacturing & 12,815 & 140,433 & 2.05 & 10.96 & 0.17 & 0.07 \\
construction & 7,270 & 46,913 & 2.16 & 6.45 & 0.14 & 0.10 \\
services & 6,319 & 94,933 & 2.12 & 15.02 & 0.33 & 0.16 \\
personal services & 3,061 & 20,991 & 2.12 & 6.86 & 0.41 & 0.18 \\
logistics & 2,562 & 36,945 & 2.28 & 14.42 & 0.12 & 0.05 \\
agriculture & 1,403 & 7,331 & 2.01 & 5.23 & 0.15 & 0.06 \\
PA & 724 & 3,082 & 2.04 & 4.26 & 0.51 & 0.32 \\
energy & 319 & 6,722 & 2.22 & 21.07 & 0.25 & 0.18 \\
mining & 244 & 1,953 & 2.28 & 8.00 & 0.11 & 0.07 \\
finance & 98 & 1,163 & 2.51 & 11.87 & 0.72 & 0.34 \\
\bottomrule
\end{tabular*}

\note[Note]{
Cross-sectional summary by industry for 2019.
Industries ranked by number of workers.
Finance pays the highest wages and employs the largest share of graduates; manufacturing and agriculture pay below average despite substantial employment shares.
}
\end{table}

\begin{table}[h!]
\caption{Mobility and connectivity}
\label{tab:connectivity}
\centering
\begin{tabular*}{.9\textwidth}[]{p{6cm}@{\extracolsep\fill}c}
\toprule
Metric & Value \\
\midrule
Workers (pre-connectivity) & 755,698 \\
Firms (pre-connectivity) & 97,981 \\
Workers (connected sample) & 751,866 \\
Firms (connected sample) & 95,665 \\
Observations (connected sample) & 3,401,702 \\
Mean firms per worker & 2.170 \\
Mean workers per firm & 16.740 \\
Share of workers with >=3 firms & 0.214 \\
Worker retention after connectivity filter & 0.996 \\
\bottomrule
\end{tabular*}

\note[Note]{
Statistics for the employer-employee mobility graph (2012-2019).
Mean firms per worker indicates average number of distinct employers observed.
The largest connected set (LCS) retains nearly all observations, indicating a well-connected labor market suitable for AKM estimation.
}
\end{table}

\begin{table}[h!]
\caption{Firm size distribution by year}
\label{tab:firm_size_distrib}
\centering
\begin{tabular*}{.9\textwidth}[]{p{2.5cm}@{\extracolsep\fill}ccccc}
\toprule
Year & Mean size & p50 & p75 & p90 & p99 \\
\midrule
2012 & 8.09 & 3.00 & 6.00 & 14.00 & 76.39 \\
2013 & 8.35 & 3.00 & 6.00 & 14.00 & 82.00 \\
2014 & 8.62 & 3.00 & 6.00 & 14.00 & 87.00 \\
2015 & 8.81 & 3.00 & 6.00 & 15.00 & 90.00 \\
2016 & 8.94 & 3.00 & 6.00 & 15.00 & 90.00 \\
2017 & 9.11 & 3.00 & 7.00 & 15.00 & 94.78 \\
2018 & 9.11 & 3.00 & 7.00 & 15.00 & 97.00 \\
2019 & 9.17 & 3.00 & 7.00 & 15.00 & 101.00 \\
\bottomrule
\end{tabular*}

\note[Note]{
Distribution of firm size (number of employees) in the estimation sample.
The median firm has 3 employees; the 99th percentile reaches 75-100 employees.
The right-skewed distribution motivates log-transformation of size variables in the predictive model.
}
\end{table}

\begin{table}[h!]
\caption{Worker age and tenure by year}
\label{tab:tenure_age}
\centering
\begin{tabular*}{.9\textwidth}[]{p{2.5cm}@{\extracolsep\fill}cccccc}
\toprule
Year & Age p25 & Age p50 & Age p75 & Tenure p25 & Tenure p50 & Tenure p75 \\
\midrule
2012 & 29.0 & 35.0 & 43.0 & 15.0 & 41.0 & 95.0 \\
2013 & 29.0 & 35.0 & 43.0 & 10.0 & 36.0 & 88.0 \\
2014 & 29.0 & 36.0 & 43.0 & 8.0 & 24.0 & 78.0 \\
2015 & 29.0 & 36.0 & 44.0 & 8.0 & 20.0 & 65.0 \\
2016 & 29.0 & 37.0 & 44.0 & 7.0 & 20.0 & 46.0 \\
2017 & 30.0 & 37.0 & 45.0 & 7.0 & 19.0 & 42.0 \\
2018 & 30.0 & 38.0 & 46.0 & 8.0 & 19.0 & 41.0 \\
2019 & 31.0 & 39.0 & 47.0 & 9.0 & 21.0 & 44.0 \\
\bottomrule
\end{tabular*}

\note[Note]{
Distribution of worker age (years) and firm tenure (months).
Median age rises from 35 to 39 over the panel, reflecting workforce aging.
Median tenure falls sharply from 41 to 21 months between 2012 and 2016, then stabilizes—consistent with increased labor market flexibility during the recovery period.
}
\end{table}

\clearpage

\subsection{Covariate definitions}\label{app:covariates}

Table~\ref{tab:covariates} reports descriptive statistics for the covariates used in TWICE.
Panel~A covers worker-level variables from QP; Panel~B covers firm-level variables from SCIE; Panel~C summarizes categorical variables.

\begin{table}[p]
\caption{Covariate descriptive statistics}
\label{tab:covariates}
\centering
{
\begin{tabular*}{\textwidth}[]{p{5.0cm}@{\extracolsep\fill}cccc}
\toprule
\multicolumn{5}{l}{\textbf{Panel A. Worker covariates}} \\
Covariate & Mean / SD & p25 & p50 & p75 \\
\midrule
Age & \makecell{37.62 \\ (9.92)} & 30.00 & 37.00 & 45.00 \\
Tenure (months) & \makecell{49.30 \\ (69.34)} & 8.00 & 23.00 & 57.00 \\
Job order (seniority) & \makecell{1.69 \\ (0.73)} & 1.00 & 2.00 & 2.00 \\
\addlinespace
\multicolumn{5}{l}{\textbf{Panel B. Firm covariates}} \\
Covariate & Mean / SD & p25 & p50 & p75 \\
\midrule
Firm age & \makecell{22.85 \\ (61.18)} & 9.00 & 18.00 & 28.00 \\
Log workers & \makecell{4.76 \\ (2.01)} & 3.14 & 4.35 & 6.01 \\
Log revenue & \makecell{15.74 \\ (2.21)} & 14.05 & 15.61 & 17.36 \\
Log revenue per worker & \makecell{10.98 \\ (1.14)} & 10.222 & 10.973 & 11.714 \\
\makecell[l]{Log revenue / \\ sales ratio} & \makecell{9.97 \\ (6.99)} & 0.000 & 12.487 & 15.535 \\
Solvency ratio & \makecell{1.75 \\ (426.78)} & 0.12 & 0.35 & 0.58 \\
Asset ratio & \makecell{0.02 \\ (0.06)} & 0.000 & 0.001 & 0.005 \\
R\&D investment & \makecell{10185.50 \\ (138362.43)} & 0.00 & 0.00 & 0.00 \\
Patent investment & \makecell{240232.98 \\ (5684569.98)} & 0.00 & 0.00 & 0.00 \\
Capital per worker (PPE) & \makecell{33622.63 \\ (161457.57)} & 1692.66 & 9451.64 & 28730.20 \\
Market concentration (\%) & \makecell{0.00 \\ (0.01)} & 0.019 & 0.058 & 0.135 \\
\addlinespace
\end{tabular*}
\begin{tabular*}{1.0\textwidth}[]{p{5.0cm}@{\extracolsep\fill}c}
\multicolumn{2}{l}{\textbf{Panel C. Categorical covariates}} \\
Covariate & Mode (share) \\
\midrule
Sector (NACE) & wholesale, retail, accommodation (0.29) \\
Education level & at most primary (0.42) \\
Qualification & specialized workers (0.57) \\
Gender & Male (0.61) \\
Legal form & Sociedade por Quotas (0.47) \\
Region (NUTS2) & Area Metropolitana de Lisboa (0.40) \\
\bottomrule
\end{tabular*}
}

\note{
Descriptive statistics for covariates used in TWICE estimation.
\emph{Panel~A (Worker covariates):} Age and tenure are continuous; job order indicates seniority rank within firm.
\emph{Panel~B (Firm covariates):} 
Log revenue per worker serves as the primary productivity measure.
Solvency ratio (equity/debt) exhibits extreme outliers (SD = 427); we use the raw variable but note that tree-based methods are robust to such skewness.
R\&D and patent investment are zero for most firms (median = 0), reflecting that innovation activity is concentrated among few firms.
Market concentration is the firm's employment share within its industry-region cell; values appear as 0.00 due to rounding but range from 0.02 to 0.11 at the quartiles.
\emph{Panel~C (Categorical covariates):} Reports modal category and its share.
Sector and region are included as categorical predictors; education and qualification define worker skill groups.
}
\end{table}

\subsection{Mobility event study}\label{app:event_study}

To assess the exogenous mobility assumption underlying the AKM benchmark, we replicate the event-study design of \cite{card2013workplace}.
Firms are classified into quartiles based on mean coworker wages (excluding the mover).
Figure~\ref{fig:chk_es} plots wage trajectories for workers transitioning between the bottom (Q1) and top (Q4) quartiles.

The results display the canonical patterns: wages are flat before the move (no pre-trends), and gains from moving up (Q1$\to$Q4) are approximately symmetric to losses from moving down (Q4$\to$Q1).
This symmetry supports the additive separability assumption of the AKM framework in our sample.

\begin{figure}[h!]
\caption{Mobility event study}
\label{fig:chk_es}
\centering
\includegraphics[width=0.9\textwidth]{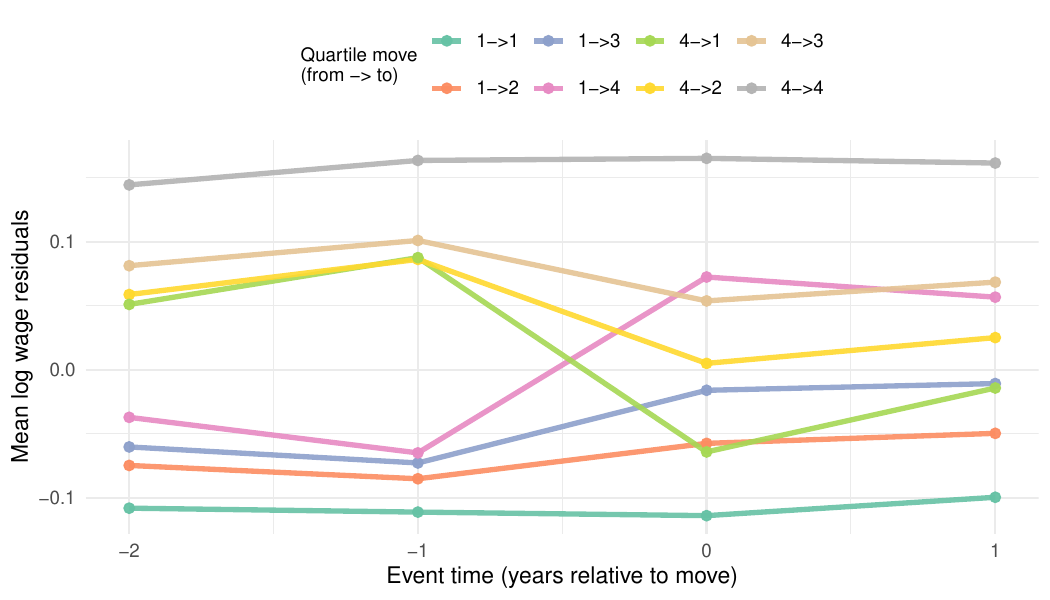}
\note[Note]{
Event study of log hourly wages around job transitions, following \cite{card2013workplace}.
Firms classified into quartiles by mean coworker wage (excluding the mover).
Year 0 is the first year at the new firm.
Flat pre-trends and symmetric wage changes for upward vs. downward moves support exogenous mobility and additivity.
}
\end{figure}

\section{TWICE estimation details}\label{a:twice}

This appendix documents implementation choices: hyperparameter selection, the two-way cross-fitting scheme, and interpretability diagnostics.

\subsection{Hyperparameters and model selection}

We select LightGBM hyperparameters, including tree depth and the number of worker- and firm-leaves (which define the cell granularity), by minimizing the two-way blocked predictive risk:
\begin{equation}\label{eq:mainloss}
\mathcal{L} = \frac{1}{B^2} \sum_{a=1}^B\sum_{b=1}^B \frac{1}{|S_{ab}|}\sum_{(i,t)\in S_{ab}}\Big(Y_{it}-\widehat f^{(-ab)}(X_{j(i,t)},Z_{it})\Big)^2
\end{equation}
This loss averages squared prediction errors over all worker–firm validation blocks and ensures robustness to the dependence structure. 
Because each leaf index corresponds to a worker or firm cell, this tuning step jointly determines model complexity and the coarseness of the partitions. 
When the true heterogeneity is close to discrete, the leaves approximate the underlying classes; when it is smoother, they provide a flexible step-function approximation.

Specifically, TWICE relies on two key granularity parameters: the number of worker classes $L$ and firm classes $K$. 
We consider a grid $\{64, 128, 256, 512\}$ for each side and select the pair $(K^\star,L^\star)$ that minimizes $\mathcal{L}$.
Table~\ref{tab:grid_KL_loss} reports the grid search results; the optimal configuration is $K=512$ firm classes and $L=512$ worker classes.

Partitioning and the final predictor use gradient-boosted trees (LightGBM) with standard regularization: learning rate 0.08, early stopping after 80 rounds without improvement, max depth 15, and minimum leaf size 30.
Categorical features are passed natively; factor levels are aligned between train and test splits.
All random draws use fixed seeds for reproducibility.

\begin{table}[t]
\caption{Blocked validation loss by grid of worker/firm cell counts}
\label{tab:grid_KL_loss}
\centering
    \begin{tabular*}{0.7\textwidth}[]{p{2.2cm}@{\extracolsep\fill}cccc}
\toprule
Firm classes $K$ & \multicolumn{4}{c}{Worker classes $L$} \\
 & 64 & 128 & 256 & 512 \\
\midrule
 64 & 0.11987 & 0.11269 & 0.10875 & 0.11297 \\
128 & 0.11713 & 0.11049 & 0.11043 & 0.10982 \\
256 & 0.12191 & 0.11167 & 0.11116 & 0.11046 \\
512 & 0.12173 & 0.11016 & 0.11015 & 0.10859 * \\
\bottomrule
\end{tabular*}

\note[Note]{
Mean squared error from two-way ID-blocked cross-validation across different numbers of firm classes ($K$, rows) and worker classes ($L$, columns).
Asterisk marks the minimum.
The optimal configuration is $K=512$, $L=512$.
}
\end{table}

\subsection{Two-way ID-blocked cross-fitting}
We respect the two-sided dependence in matched data by blocking on \emph{IDs} rather than on observations. 
The training sample (restricted to the connected set) is partitioned into $B=5$ disjoint blocks of worker IDs and $B=5$ disjoint blocks of firm IDs, inducing $B^2=25$ validation cells $\{S_{ab}\}_{a,b=1}^B$. 
For each $(a,b)$ we train the predictor on the complement of $S_{ab}$—i.e., excluding \emph{all} rows involving any worker in block $a$ or any firm in block $b$—and we score only $S_{ab}$. 
No worker or firm ever appears in both the training data and the held-out cell used to score that worker–firm observation. 
The cross-fitted risk used for model selection is the average MSE across all cells as per equation~\eqref{eq:mainloss}.

\subsection{External test set}

For an untouched benchmark, we draw a firm-level holdout from the connected set and sample workers within those firms. 
All rows linked to held-out firms are excluded from model fitting and tuning (including cross-fitting folds). 
We report MSE and $R^2$ on both train and test; $R^2$ is the squared correlation between $Y$ and $\widehat Y$.

\subsection{Interpretability diagnostics}

TWICE produces three types of model diagnostics: variable importance, partial dependence plots (PDPs), and accumulated local effects (ALEs).
These tools summarize which covariates drive the partitions and how they relate to predicted wages.

\paragraph{Variable importance}
Figure~\ref{f:rel_importance} reports LightGBM variable-importance measures (``average gain'') for the three TWICE models.
Importance is the total reduction in the loss function attributed to splits on each feature, normalized to sum to one.

\begin{figure}[t]
\caption{Variable importance}\label{f:rel_importance} 
\subcaptionbox{Firm partition\label{subf:firm_importance}}{\includegraphics[scale=0.43]{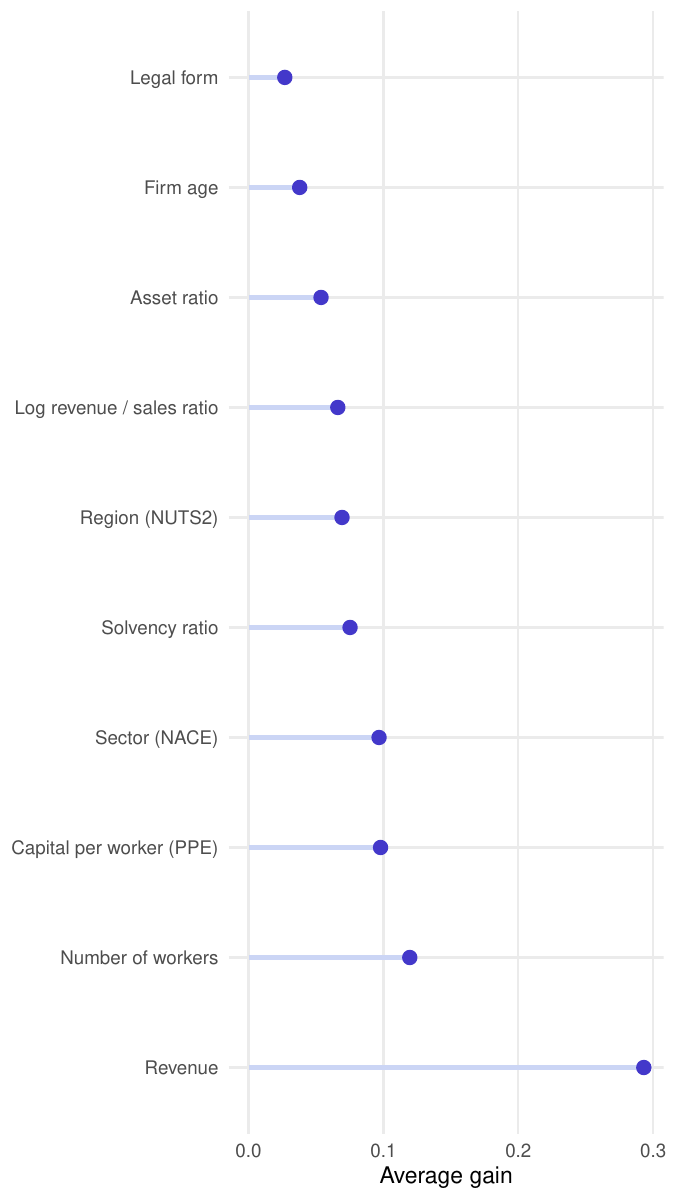}}\hfill
\subcaptionbox{Worker partition\label{subf:worker_importance}}{\includegraphics[scale=0.43]{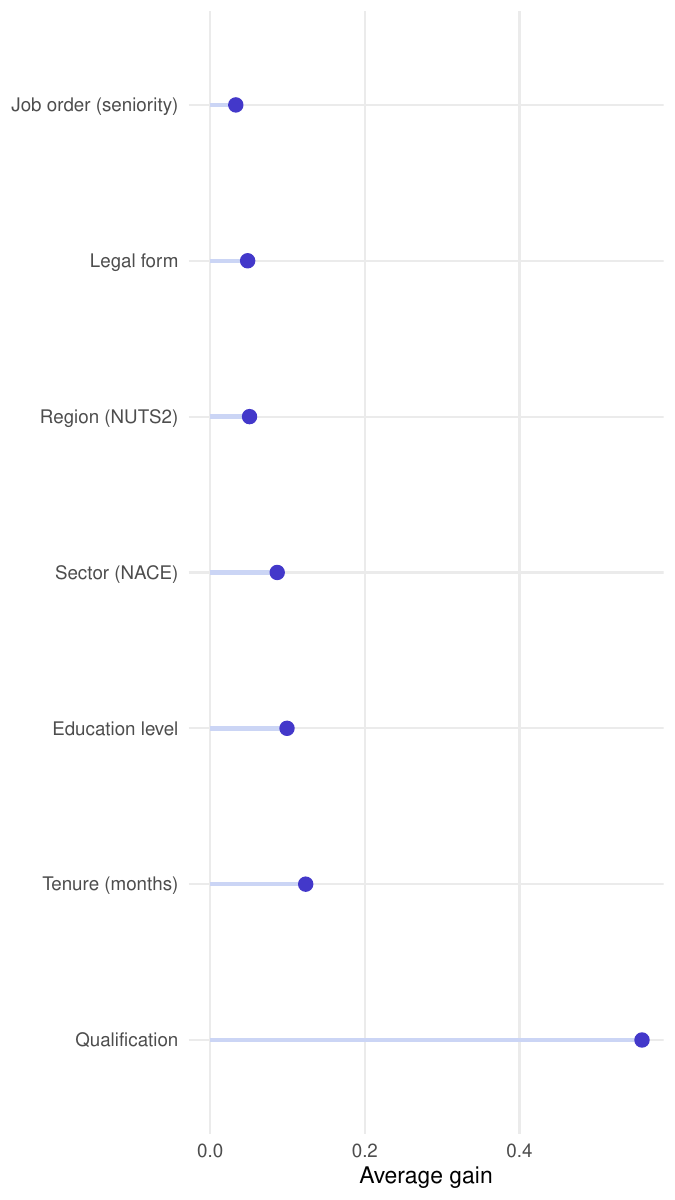}}\hfill
\subcaptionbox{Conditional wage model\label{subf:worker_firm_importance}}{\includegraphics[scale=0.43]{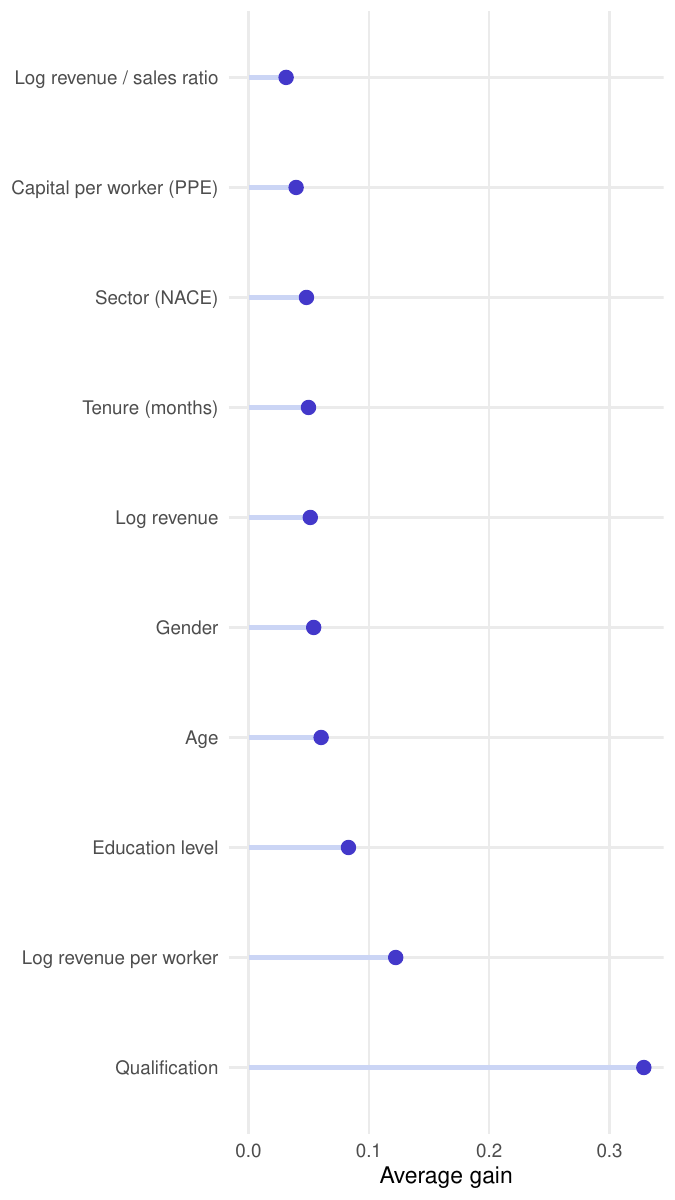}}
\note{
Relative variable importance for TWICE classification trees.
Panel~A: features used to partition firms into classes $h(j)$.
Panel~B: features used to partition workers into classes $g(i)$.
Panel~C: features used in the final conditional wage predictor.
}
\end{figure}

Panel~A (firms) shows that log revenue is the most important predictor, followed by sector, number of workers, region, capital per worker, and financial ratios (revenue/sales, solvency). 
These variables account for most splits used to classify firms into wage-relevant groups.

Panel~B (workers) indicates that qualification, tenure, and education are the main predictors of worker classification. 
Sector, region, legal form, and job seniority contribute additional separation.

Panel~C (conditional wage model) shows that qualification and education remain important, but log revenue and age enter prominently once firm characteristics are included. 
Gender, tenure, and capital per worker play smaller roles.

\paragraph{PDP implementation}\label{app:pdps}
For a focal feature $v_S$, we evaluate a 40-point quantile grid (trimmed to the 10th--90th percentiles).
For each grid value $s$, we construct a modified dataset by replacing the focal coordinate $V_S$ by $s$ for every observation, compute predictions, and average across observations. 
Predictions are then averaged across the 25 cross-fitted models.

For the firm-size curve, we report a \emph{conditional} PDP: in addition to setting firm size to $s$, we fix log revenue per worker at its sample median for every observation, while leaving all remaining covariates at their observed values.

\paragraph{ALE implementation}
ALEs complement PDPs when predictors are correlated. 
For a continuous feature $x$, we partition its support into 40 grid points, compute average finite differences within each bin conditional on observed $X_{-x}$, and accumulate across bins; the curve is mean-centered by construction. 
Display is restricted to the 10th--90th percentile range of each feature to avoid extrapolation.
For tenure, which has a highly right-skewed distribution, we display the x-axis on a log scale.

\subsection{Summary of the estimation procedure}

The estimation routine is summarized in Algorithm \ref{alg:twice}.

\begin{algorithm}[t]
\caption{TWICE Estimation Procedure}
\label{alg:twice}
\begin{algorithmic}[1]
\Require Matched worker–firm data $\{Y_{it},X_j,Z_{it}\}$ restricted to the largest connected set; candidate grids $\mathcal K$ and $\mathcal L$ for the numbers of firm and worker classes
\For{each $(K,L)\in \mathcal K \times \mathcal L$}
  \State \emph{Firm partition:} estimate a supervised tree-based regression of $Y_j$ on $X_j$ with at most $K$ leaves; let $\kappa_K(j)\in\{1,\dots,K\}$ be the induced firm classes
  \State \emph{Worker partition:} estimate a supervised tree-based regression of $Y_{it}$ on $Z_{it}$ with at most $L$ leaves; let $\lambda_L(i)\in\{1,\dots,L\}$ be the induced worker classes
  \State \emph{Wage prediction:} estimate the conditional wage function using two-way cross-fitting on worker and firm IDs, with regressors $(X_j,Z_{it},\kappa_K(j),\lambda_L(i))$
  \State \emph{Model selection:} compute the blocked out-of-sample loss $\mathcal L(K,L)$ as in~\eqref{eq:mainloss}
\EndFor
\State Select $(K^\ast,L^\ast)=\arg\min_{K,L}\mathcal L(K,L)$
\State Re-estimate the final model on the training portion of the connected set using $(K^\ast,L^\ast)$ and evaluate on a separate firm-level holdout sample
\end{algorithmic}
\end{algorithm}

\section{Robustness of the firm partition\label{app:robustness_firm}}

This appendix assesses whether the TWICE variance decomposition is sensitive to how the firm partition is constructed in the first-stage supervised learning step. 
In the baseline approach, firm-year types are obtained by training a LightGBM model that predicts a firm-year wage target using observable firm characteristics, and defining firm classes by the model’s leaf assignments (holding the tuned leaf-count fixed). 
Since firm-year average wages may reflect both firm pay policies and the composition of workers employed at the firm, we consider alternative targets designed to reduce the role of workforce composition.

\subsection{Alternative targets}
Let $W_{it}$ denote the hourly wage measure used in the first-stage firm model, and let $j = J(i,t)$ index the firm employing worker $i$ in year $t$.
We consider three firm-year targets:

\paragraph{(i) Baseline (mean target)}
The firm tree is trained using the firm’s average log wage, 
\[ \bar W_{jt}^{\text{mean}} = \frac{1}{n_{jt}}\sum_{i,t:J(i,t)=j} W_{it} \]

\paragraph{(ii) Median target}
\[ \bar W_{jt}^{\text{med}} = \text{median} \bc{W_{it} : J(i,t) = j} \]
This reduces sensitivity to outliers and extreme observations while still using a firm-year wage summary.

\paragraph{(iii) Residual target (worker-only residual aggregation)}
We first estimate a worker-only wage predictor $\widehat m(X_{it})$ using a cross-fitted LightGBM model based on worker characteristics only (in the current implementation: year, tenure, age, sex, education, qualification, sector, and region). Using cross-fitting by worker groups, we compute residuals
\[ \widehat \varepsilon_{it} = W_{it} - \widehat m_{-k(i)}(X_{it}), \]
where $\widehat m_{-k(i)}$ is trained excluding the fold containing worker $i$.
We then aggregate residuals to the firm-year level:
\[ \bar W_{jt}^{\text{res}} \frac{1}{n_{jt}} \sum_{i : J(i,t) = j} \widehat \varepsilon_{it} \]
The firm model is then trained to predict $W_{jt}^{\text{res}}$ using the same firm covariates and the same tuned leaf count as in the baseline.

Across all three specifications, we hold the worker partition fixed and only retrain the firm partition under the alternative firm-year targets, then recompute the TWICE decomposition.

\paragraph{Results}

\begin{table}[t]
\caption{Robustness}
\label{tab:robustness}
\centering
\begin{tabular*}{.9\textwidth}[]{p{3.8cm}@{\extracolsep\fill}ccccc}
\toprule
Specification & Worker & Firm & Sorting & Interaction & Residual \\
\midrule
Baseline (mean) & 0.276 & 0.087 & 0.116 & 0.073 & 0.448 \\
Median target & 0.275 & 0.091 & 0.117 & 0.075 & 0.443 \\
Residual target & 0.317 & 0.086 & 0.085 & 0.079 & 0.433 \\
\bottomrule
\end{tabular*}

\note[Note]{
Each entry reports the share of total variance of log wages attributed to the TWICE worker component, firm component, sorting term, interaction term, and the within-cell residual.
Across specifications we keep the worker partition fixed and retrain the firm partition using alternative first-stage targets (mean, median, or worker-residual firm-year aggregates), then recompute the TWICE decomposition under the resulting firm classes.
}
\end{table}

Table \ref{tab:robustness} reports the TWICE variance-share decomposition under the three firm targets. 
The baseline and median targets deliver nearly identical decompositions. 
The residual-target specification mechanically shifts variance away from the firm and sorting components and toward the worker component, consistent with the fact that the residual target is constructed to remove worker-composition-driven variation from the firm-year target used to define firm classes. 
Importantly, the qualitative structure of the decomposition remains: there is still a meaningful firm component, a sorting term, and a non-trivial interaction component.

\begin{table}[t]
\caption{Correlation}
\label{tab:correlation}
\centering
\begin{tabular*}{0.6\textwidth}[]{p{7cm}@{\extracolsep\fill}c}
\toprule
Metric & Value \\
\midrule
Spearman corr. (firm component psi) & 0.752 \\
\bottomrule
\end{tabular*}

\note[Note]{
Spearman correlation comparing the baseline (mean-target) and residual-target specifications. The metric is computed on the firm-year firm component $\psi$ implied by the TWICE additive projection under each partition.
}
\end{table}

Table \ref{tab:correlation} reports the Spearman correlation between the firm-year firm component $\psi$ implied by the baseline partition and by the residual-target partition. 
The correlation is sizeable, indicating that—even when the firm partition is trained on a worker-residual target—the resulting firm pay component remains broadly aligned with the baseline ranking of firm heterogeneity captured by TWICE.

Overall, these results show that the firm partition—and the resulting variance decomposition—is not driven by the specific choice of firm-level target, and the main  empirical conclusions of the paper are robust to alternative constructions of firm types.

\end{document}